\documentclass[8pt,reqno]{amsart}
\usepackage{latexsym}
\usepackage{amssymb}
\usepackage{amsmath}
\usepackage{amsthm}
\usepackage{amsfonts}
\usepackage{amssymb, epic, xypic} % for the diagrams!
\usepackage{graphicx}
\usepackage{booktabs}

\numberwithin{equation}{section}

\newcommand{\diag}{{\rm diag}}
\newcommand{\jj}{\mathfrak{J}}
\newcommand{\hh}{\mathbb{H}}
\newcommand{\g}{\mathfrak{g}}
\newcommand{\h}{\mathfrak{u}}
\newcommand{\p}{\mathfrak{p}}
\newcommand{\lie}{{\rm Lie}}
\newcommand{\RR}{\mathbb R}
\newcommand{\dd}{{\rm dim}}
\newcommand{\tr}{{\rm tr}}
\newcommand{\Oct}{\mathbb{O}}
\newcommand{\cc}{\mathbb{C}}
\newcommand{\rr}{\mathbb{R}}
\newcommand{\zz}{\mathbb{Z}}

\DeclareMathOperator{\SL}{SL}
\DeclareMathOperator{\sla}{sl}
\DeclareMathOperator{\SO}{SO}

\DeclareMathOperator{\SU}{SU}
\DeclareMathOperator{\su}{su}
\DeclareMathOperator{\trace}{Tr}
\DeclareMathOperator{\Ad}{Ad}

\DeclareMathOperator{\Aut}{Aut}

\newcommand{\transp}[1]{{}^t\!#1}
\newcommand{\smaq}{\left[ \begin{smallmatrix}}
\newcommand{\smat}{\left( \begin{smallmatrix}}
\newcommand{\smcq}{\end{smallmatrix}\right]}
\newcommand{\smct}{\end{smallmatrix}\right)}
\newcommand{\smag}{\left \{ \begin{smallmatrix}}
\newcommand{\smcg}{\end{smallmatrix}\right \}}

\newcommand{\oo}{0_{27}}
\newcommand{\ov}{\vec 0_{27}}

\newcommand{\be}{\begin{equation}}
\newcommand{\ee}{\end{equation}}
\newcommand{\bes}{\begin{equation*}}
\newcommand{\ees}{\end{equation*}}

\newcommand{\eqn}{\begin{eqnarray}}
\newcommand{\feqn}{\end{eqnarray}}
\newcommand{\eqnn}{\begin{eqnarray*}}
\newcommand{\feqnn}{\end{eqnarray*}}

\newtheorem{prop}{Proposition}

\begin{document}

\title[$E_7$ groups from octonionic magic square]{\boldmath{$E_7$} groups from octonionic magic square}
\author{Sergio L. Cacciatori}
\author{Francesco Dalla Piazza}
\address{
Dipartimento di Fisica e Matematica, Universit\`a degli Studi dell'Insubria,
Via Valleggio 11, 22100 Como, Italy, and INFN, via Celoria 16, 20133 Milano, Italy}
\email{sergio.cacciatori@uninsubria.it}
\email{f.dallapiazza@uninsubria.it}

\author{Antonio Scotti}
\address{Dipartimento di Matematica,
Universit\`a degli Studi di Milano,  Via Saldini 50, 20133 Milano, Italy}
\email{ascotti@mindspring.com}

%\keywords{}
%\subjclass{}
%\maketitle

\begin{abstract}
In this paper we continue our program, started in \cite{Cacciatori:2005yb}, of building up explicit
generalized Euler angle parameterizations for all exceptional compact Lie groups. Here we solve the problem for $E_7$, by first
providing explicit matrix realizations of the Tits construction of a Magic Square product between the exceptional octonionic algebra $\jj$ and the quaternionic
algebra $\hh$, both in the adjoint and the $56$ dimensional representations. Then, we provide the Euler parametrization of $E_7$ starting from
its maximal subgroup $U=(E_6\times U(1))/\mathbb{Z}_3$. Next, we give the constructions for all the other maximal compact subgroups.

\end{abstract}

\maketitle

%%%%%%%%%%%%%%%%%%%%%%%%%%%%%%%%%%%%%%%%%%%%%%%%%%%%%%%%%%%%%%%%%%%%%%%%%%%%%%%%%%%%%%%%%%%%%%%%%
\section{Introduction}
Simple Lie groups are well understood, and find application in a huge number of mathematical and physical
fields. In particular,  the role of compact exceptional Lie groups in grand unification gauge
theories and in string theories, and that of the noncompact forms in supergravity theories is well known.
Many properties of these groups can be already inferred from abstract
theoretical considerations, however it can be useful to have available  explicit concrete realizations of such groups in term of matrices,
for testing conjectures
related to confinement in non abelian gauge theories, doing explicit nonperturbative computations in exceptional lattice GUT theories
and in random matrix theories. Beyond these, there are other useful applications in physics or mathematical physics of an explicit matrix realization of
the $E_6$, $E_7$ and $E_8$ exceptional Lie groups. For example: sigma models based on exceptional Lie group quotients are of interest for string theory
and conformal field theory applications; the study of the properties of the magnetic material Cobalt Niobate is also based on exceptional Lie groups of
type $E$ \cite{cobalt}.  The connection to explicit realizations and special function theory would permit to perform calculations
of matrix elements. These applications are also directly interesting in integrable models.
In particular there is a specific motivations from physics to be interested to $E_7$ among all exceptional groups:
recently a strict relation between cryptography and black hole physics based on  E7 exceptional supergravity has been discovered.
However, actual computation of entangled expectation values require again explicit determination of the Haar measure and range of parameters.
Moreover, the most general structure of the attractor mechanism of black holes in $N=2$, four dimensional supergravity arises in $E_7$ exceptional supergravity.
The orbit of the U-duality group are only partially known just because a suitable explicit realization of the group $E_7$ is lacking.

\

In this paper we will focus on the compact form of $E_7$. In this case
the main difficulty consists in finding a realization admitting a simple
characterization of the range of parameters.
A way, mainly inspired by \cite{Tilma} and \cite{Sudarshan},
to solve this problem has been introduced in \cite{Cacciatori:2005yb} for the exceptional
Lie group $G_2$ (see also \cite{Cacciatori:2005gi}), then applied to the $SU(N)$ groups in \cite{Bertini:2005rc}
and to the exceptional Lie groups $F_4$ and $E_6$ in \cite{F4} and \cite{E6} respectively. In this
paper we continue our program of building up the generalized Euler parametrization for all the
exceptional Lie groups. There are many possible realizations of the Euler parametrization on a given group (see \cite{40} for a review),
depending on the choice of the maximal subgroup one starts with. In any case, the simplest one is that based on
the higher dimensional compact subgroup. For $E_7$, this is the group $U=(E_6\times U(1))/\mathbb{Z}_3$, where $\mathbb{Z}_3$
is the center of $E_6$, which is indeed the first case we consider here. The other possible maximal compact subgroups of $E_7$ are
$SU(8)/\mathbb{Z}_2$ and $({\rm Spin}(12)\times SU(2))/(\mathbb{Z}_2 \times \mathbb{Z}_2)$ associated to the real forms $E_{7(7)}$ and
$E_{7(-5)}$ respectively. We will provide a construction of the compact real form for each choice of the maximal compact subgroup. Indeed, in practical
applications it can be necessary to select a specific subgroup among the others. Moreover, in other applications, as for example
in extended supergravities, it becomes important to determine the various real forms and the corresponding symmetric spaces. In our
constructions these can be easily realized by applying the unitary Weyl trick.

\

The plan of the paper is as follows. In Sec. \ref{sec:e7} we present the main features of the Lie group $E_7$.
In Sec. \ref{sec:tits} we review the Tits construction for the Lie
algebra of $E_7$. In particular we derive from it the main properties we will need to apply the generalized
Euler angles method. We will present both the adjoint representation ${\bf 133}$ and the smallest fundamental representation ${\bf 56}$.
We build up the group, presenting a very careful exposition of details in the appendices.\\
In Sec. \ref{sec:E7(7)} we give a second construction associated to the split form $E_{7(7)}$, and in Sec. \ref{sec:E7(-5)} a third
construction based on the subgroup $({\rm Spin}(12)\times SU(2))/(\mathbb{Z}_2 \times \mathbb{Z}_2)$.
In these two cases we will not present all necessary checks, as, for example, the explicit computation of the volumes and the related
integrals, but we limit ourselves to indicate the main steps. Indeed, the lacking details can be included in a more general framework
which deserves to be presented apart \cite{to-appear}.
%%%%%%%%%%%%%%%%%%%%%%%%%%%%%%%%%%%%%%%%%%%%%%%%%%%%%%%%%%%%%%%%%%%%%%%%%%%%%%%%%%%%%%%%%%%%%%%%
\subsection*{Remark}
In www.dfm.uninsubria.it/E7/ one can find the Mathematica programs providing the constructions of the matrix realizations
of the $\bf{133}$ and $\bf{56}$ algebra representations for the first construction, and the $\bf{56}$ for the other ones.

%%%%%%%%%%%%%%%%%%%%%%%%%%%%%%%%%%%%%%%%%%%%%%%%%%%%%%%%%%%%%%%%%%%%%%%%%%%%%%%%%%%%%%%%%%%%%%%%%
\section{The exceptional Lie group $E_7$} \label{sec:e7}
As a complex Lie algebra, $E_7$ is the unique exceptional Lie algebra of rank $7$. It is characterized by the Dynkin diagram drawn in Fig. \ref{fig}.
\begin{figure}[h]
\centering
\caption{Dynkin diagram for $E_7$}\label{fig}
\begin{picture}(200, 90)(-50,-70)%
\put(-110,10){\makebox(0,0)[l]{{$\alpha_1$({\bf 133})}}}
\put(-60,10){\makebox(0,0)[l]{{$\alpha_3$({\bf 8645})}}}
\put(-10,10){\makebox(0,0)[l]{{$\alpha_4$({\bf 365750})}}}
\put(40,10){\makebox(0,0)[l]{{$\alpha_5$({\bf 27664})}}}
\put(90,10){\makebox(0,0)[l]{{$\alpha_6$({\bf 1539})}}}
\put(140,10){\makebox(0,0)[l]{{$\alpha_7$({\bf 56})}}}
\put(-20,-60){\makebox(0,0)[l]{{$\alpha_2$({\bf 912})}}}
\put(-100,0){\circle*{6}}
\put(-50,0){\circle*{6}}
\put(0,0){\circle*{6}}
\put(50,0){\circle*{6}}
\put(100,0){\circle*{6}}
\put(150,0){\circle*{6}}
\put(0,-50){\circle*{6}}
%\drawline(0,-120)(0,120)
%\drawline(-120,0)(120,0)
\thicklines
\drawline(-100,0)(150,0)
\drawline(0,0)(0,-50)
%\drawline(0,0)(50,0)
%%%%%%%%%%%%%%%%%%%%%%%%%%%%%%%%%%%%%%%%%%%
\end{picture}
\end{figure}
Recall that to each dot corresponds a simple root $\alpha_i$. These are free generators of the root lattice $\Lambda_R=\sum_i \mathbb{Z} \alpha_i$. The
space $H^*=\Lambda_R \otimes \mathbb{R}$ is endowed with a positive definite inner product $(|)$. The weight lattice $\Lambda_W$ is the dual of $\Lambda_R$
w.r.t. the hooked product, which means that it is freely generated over $\mathbb{Z}$ by the fundamental weights $\lambda^i\in H^*$, $i=1,\ldots,7$
defined by $\langle \alpha_i , \lambda^j \rangle=\delta_i^j$, with
\begin{equation}
  \langle \alpha, \lambda \rangle := 2 \frac {(\alpha|\lambda)}{(\alpha|\alpha)}.
\end{equation}
To each fundamental weight corresponds univocally a fundamental representation, and all irreducible finite dimensional representations can be generated
from the fundamental ones. In Fig. \ref{fig} the fundamental representations are indicated in parenthesis. Here, we will deal with the two lower dimensional
ones.\\
We are interested in constructing the compact form of the group $E_7$. For this reason it is worth to mention some further facts about both the $E_7$ algebras
and groups. As we said, the complex algebra $E_7$ is completely characterized by the Dynkin diagram, from which one can reconstruct the adjoint representation
that, being faithful, is isomorphic to the algebra itself. This is a $133$ dimensional algebra that, indeed, correspond to the fundamental representation
${\bf 133}$. This algebra contains four distinct real forms. This means that there are four inequivalent ways to select a $133$ dimensional real
subspace of the $266$ dimensional real space underlying the complex algebra $E_7$, in such a way that the selected subspace inherited with the Lie
product is itself a (real) Lie algebra. For each simple Lie algebra $A$ there is a unique simply connected Lie group $G$ (up to isomorphisms) such that
$A$ is the corresponding Lie algebra. The complex Lie group $E_7$ contains a maximal compact subgroup which is a $133$ dimensional real Lie group
whose Lie algebra is then called the compact form and usually indicated\footnote{The Killing form $K$ on a complex Lie algebra
is defined by $K(X,Y):=\tr (ad(X) ad(Y))$ and is non degenerate for a simple algebra and on the corresponding real forms. In
particular, for a compact form it is negative definite on the maximal compact subalgebra, the maximal Lie subalgebra that exponentiated generates
a compact group.} $E_{7(-133)}$, where in parenthesis the signature of the Killing form (number of positive eigenvalues minus
number of the negative ones) is indicated. We will use the same notation to indicate both the group and the algebra.
When referring to the group
we will mean the unique simply connected group.\\
The compact forms are in correspondence with the maximal compact subgroups of $E_{7(-133)}$, the compact subgroups that are not properly contained
in a proper subgroup. There are four of such subgroups and then four real forms, which we collect in table \ref{tab}.
\begin{table}[hbtp]
\begin{center}
%\resizebox*{1\textwidth}{!}{
\begin{tabular}{ccc}
\toprule
Symbol & Real Form & Maximal compact subgroup  \\
\midrule
$E_{7(-133)}$ & Compact & $E_{7(-133)}$ \\
$E_{7(7)}$ & Split & $SU(8)/\mathbb{Z}_2$ \\
$E_{7(-5)}$ & EVI & $(Spin(12)\times Sp(1))/\mathbb{Z}_2$ \\
$E_{7(-25)}$ & EVII & $(E_{6(-78)}\times U(1))/\mathbb{Z}_2$ \\
\bottomrule
\end{tabular}
\caption{The real forms of $E_7$.}
\label{tab}
\end{center}
\end{table}
\noindent
For a given real form we can write the corresponding Lie algebra as $A=T+P$, where $T$ is the maximal compact subalgebra and $P$ its complement.
From this, one obtains the compact form by using the Weyl unitary trick $P\mapsto iP$. Here we are interested in the compact form only. Nevertheless we
will construct it in three different ways, each one evidencing a different proper maximal subalgebra. In this way the interested reader can
reobtain the corresponding noncompact real form directly by applying the Weyl unitary trick.\\
Finally, we stress that for the group $E_{7(-133)}$ we mean the simply connected group, which is the universal covering group. In general, all Lie groups having
the same (finite dimensional real) Lie algebra are obtained from the universal covering by quotienting with a discrete subgroup of the center. For $E_7$
the center is $\mathbb{Z}_2$. Note that the adjoint representation of the Lie algebra is faithful, however, the same is not true for the Adjoint
representation of the universal cover group $G$, since the center $C$ is just the kernel of the Adjoint map. Then, in general, the Adjoint representation
of the group realizes the group $G/C$ in place of $G$. Instead, a faithful representation is obtained by exponentiating the lowest dimensional
fundamental representation.\footnote{The unique simple Lie group whose lowest dimensional representation coincides with
the adjoint one is the exceptional group $E_8$. Since in this case the center is trivial, our statement remains true.}
In our case this means that we need to exponentiate the ${\bf 56}$. Nevertheless, to get more information in some case we will work out the adjoint
representation also.

%%%%%%%%%%%%%%%%%%%%%%%%%%%%%%%%%%%%%%%%%%%%%%%%%%%%%%%%%
\section{$E_{7(-25)}$ construction}\label{sec:tits}
\subsection{The Tits construction and {\bf 133}}\label{sec:tits1}
We start with the construction of the adjoint representation of the algebra. In order to catch the idea, recall first that the exceptional
Lie group $F_{4(-52)}$ can be realized as the group of automorphisms of the exceptional Jordan algebra $\jj_3{\Oct}$.
Its Lie algebra is the Lie algebra of derivations over $\jj:=\jj_3(\Oct)$. Thanks to a result due to Chevalley and Schafer, this algebra can be
extended to an $E_6$ type algebra (more precisely the $E_{6(6)}$ split form) by adding the action of the traceless part
$\jj'$ over $\jj$ itself, naturally given by Jordan multiplication:
\begin{eqnarray}
  E_{6(6)}=D(\jj) \oplus \jj',
\end{eqnarray}
where the symbol $D()$ means ``the linear derivations of''.
To obtain the compact form from the split form one has to apply the Weyl trick to $\jj'$ that is to say that we have to
``complexify'' the Jordan algebra by adding the imaginary part $i\otimes \jj'$. \\
This way to construct the algebra can be summarized by saying that $E_6$ is the Magic Square product between $\jj$ and $\mathbb{C}$.
This can be extended to the realization of
the $E_7$ compact form as the Magic Square product between the exceptional octonionic algebra $\jj$ and the quaternionic
algebra $\hh$, \cite{bart-sud} \cite{tits} \cite{Schaf}. We will refer to it as the Tits construction.
In our language, this means that we need to ``quaternionize'' the Jordan algebra so that the vector space underlying the resulting algebra will be
\eqn
\mathfrak{g}=D(\hh)\oplus D(\jj) \dot + (\hh' \otimes \jj'),\label{nonalg}
\feqn
where $\jj'$ is the set of traceless octonionic Jordan matrices and $\hh'$ are
the imaginary quaternions. Here we use the symbol $\oplus$ to mean direct sum of algebras whereas $\dot +$ is a direct sum of vector
spaces but not of algebras. This is in order to stress that we have not yet extended the Lie product to the last summand.
Thus, $D(\hh)\oplus D(\jj)$ is a subalgebra of $\mathfrak{g}$. Moreover, there is a natural action of
$D(\hh)\oplus D(\jj)$ over $\hh' \otimes \jj'$ which defines the mixed product
\eqn
[(H,J), h\otimes j]=H(h)\otimes j+h \otimes J(j) \qquad\ \forall (H,J)\in D(\hh)\oplus D(\jj), \quad h\otimes j\in \hh' \otimes \jj'.
\feqn
To extend (\ref{nonalg}) to a Lie algebra one must define a product between elements of $\hh' \otimes \jj'$. This requires the introduction of
some notations and properties.
%%%%%%%%%%%%%%%%%%%%%%%%%%%%%%%%%%%%%
\subsection*{Geometry of quaternions.}
On $\hh$ an inner product is defined $\langle h_1,h_2 \rangle={\rm Re} (\bar h_1 h_2)$, where complex conjugation changes the sign of the
imaginary units: if $h=A_0+iA_1+jA_2+kA_3$, $A_i\in \mathbb{R}$, then $\bar h=A_0-iA_1-jA_2-kA_3$. From $h_1$ and $h_2$ one defines the derivation
$D_{h_1,h_2}\in D(\hh)$ by
\eqn
D_{h_1,h_2}=[L_{h_1}, L_{h_2}]+[R_{h_1}, R_{h_2}],
\feqn
where $L$ and $R$ are the usual left and right translations.
It is convenient to fix the orthonormal basis $h_0=1$, $h_i$, $i=1,2,3$, where $h_1=i, h_2=j, h_3=k$ are
the imaginary units of $\hh$.
A basis for $D(\hh)$ is thus given by $H_i=ad_{h_i}$, $i=1,2,3$.
%%%%%%%%%%%%%%%%%%%%%%%%%%%%%%%%%%%%
\subsection*{Geometry of the Jordan algebra}
On $\jj$ we can define the inner product $\langle j_1,j_2 \rangle= {\rm Tr} (j_1\circ j_2)$, where $\circ$ is the Jordan product
$j_1\circ j_2=(j_1j_2 +j_2 j_1)/2$. The subspace of matrices orthogonal to the $3\times 3$
identity $I_3$
is thus $\jj'$. On it we can define the product $\star : \jj' \times \jj' \to \jj'$ defined by
\eqn
\star :(j_1,j_2)\mapsto j_1\star j_2=j_1 \circ j_2 -\frac 13 \langle j_1,j_2 \rangle I_3.
\feqn

\

\noindent With these simple tools in mind one can complete the Lie product on (\ref{nonalg}) by defining
\eqn
[h_1\otimes j_1, h_2\otimes j_2]:= \frac {\sigma^2}{12} \langle j_1,j_2 \rangle D_{h_1,h_2}-\sigma^2 \langle h_1,h_2\rangle [L_{j_1}, L_{j_2}]
+\frac \sigma2 [h_1,h_2]\otimes (j_1\star j_2) \label{prod}
\feqn
in order to obtain a Lie algebra. For simplicity, we will use the notation
\eqn
E_{7(-133)}=D(\hh)\oplus D(\jj) \oplus (\hh' \otimes \jj'),\label{alg}
\feqn
to indicate the resulting Lie algebra, and choose $\sigma=1$.\\
A sketch of the proof of the validity of (\ref{prod}) can be found in appendix \ref{app:a}. From these rules one easily reconstruct explicitly a
matrix realization of the adjoint representation. This is done in appendix \ref{sec:adjoint}.\\
Before going to the construction of the smallest representation {\bf 56} it is worth to include some comments on the subalgebras.

%%%%%%%%%%%%%%%%%%%%%%%%%%%%%%%%%%%%%%%%%%%%%%%%%%%%%%%%%%%%%%%%%%%%%%%%%%%%%%%
\subsection*{The $F_4$ and $E_6$ subalgebras and some useful relations}
The $F_4$ Lie subalgebra is manifestly included as the algebra of derivation over the exceptional real octonionic Jordan algebra $\jj$.
If we choose a basis $\{j_a\}_{a=1}^{26}$ of $\jj'$ normalized by $\langle j_a,j_b \rangle=\tau \delta_{ab}$, with a suitable real $\tau$, then
it can be completed to a basis for $\jj$ by adding $j_{27}=\sqrt {\frac \tau3} I_3$. After fixing this, we can determine a basis
$\{D_I\}_{I=1}^{52}$ for $D(\jj)$ as in \cite{F4}.
A 27-dimensional matrix representation is then obtained as
\eqn
(C_I)^{\mu}_{\phantom{\mu}\nu}=\hat j^\mu (D_I (j_\nu)), \label{effe4}
\feqn
where $\{\hat j^\mu\}_{\mu=1}^{27}$ is the dual basis $\hat j^\mu (j_\nu)=\delta^\mu_\nu$. Since  $D_I (j_{27})=0$, the last row and column
of all matrices vanish, so that, this is a ${\bf 26}\oplus {\bf 1}$ representation. Its extension to the 27 fundamental representation of the $E_6$
algebra is obtained by adding the operators corresponding to the right multiplication by $\jj'$. This adds 26 $27\times 27$ matrices
$(\tilde C_a)^\mu_{\phantom{\nu}\nu}$ defined by
\eqn
(\tilde C_a)^\mu_{\phantom{\mu}\nu}=-i\ \hat j^\mu (R_{j_a} (j_\nu)),\qquad
a=1,\cdots,26,\;\;\mu,\nu=1,\cdots,27. \label{e6}
\feqn
The factor $-i$ has been introduced in order to implement the Weyl trick.
Note that this factor consists in making a choice among all possible imaginary quaternions. Passing from one to another choice is done by
the acting with the $SU(2)$ group symmetry. With our choice we break this symmetry down to the $U(1)$ subgroup that lives the imaginary unit invariant.
Its Lie algebra is generated by the derivation $D_i$ over the quaternions, which indeed commutes with the whole $E_6=D(\jj) \oplus \jj'$. In this
way we have recognized the maximal compact subalgebra $E_6 \times \mathbb{R}=\lie((E_6\times U(1))/\mathbb{Z}_3)$.\\
It is interesting to note that, $\jj$ being Abelian, we have $R_{j_a} (j_b)=R_{j_b} (j_a)$ which implies
\eqn
(\tilde C_a)^\mu_{\phantom{\nu} b}=(\tilde C_b)^\mu_{\phantom{\nu} a}. \label{symm}
\feqn
We will see in the next subsection that this symmetry relation is  a particular case of a more general symmetry relation
that has a deep geometrical meaning.\\
Other interesting relations are obtained from the Leibnitz property $D_I(j_a\circ j_b)=D_I(j_a) \circ j_b+j_a \circ D_I(j_b)$:
it becomes
\eqn
(\tilde C_a)^c_{\phantom{c}b} (C_I)^d_{\phantom{c}c} j_d=[(C_I)^c_{\phantom{c}a} (\tilde C_c)^\mu_{\phantom{c}b}
+(C_I)^c_{\phantom{c}b} (\tilde C_a)^\mu_{\phantom{c}c}]j_\mu.
\feqn
Applying the dual basis $\hat j^\nu$ we get
\eqn
&& [C_I, \tilde C_a]^d_{\ b}=(C_I)^c_{\ a} (\tilde C_c)^d_{\ b},\\
&& (C_I)^c_{\ a} (\tilde C_b)^{27}_{\ \ c}+(C_I)^c_{\ b} (\tilde C_a)^{27}_{\ \ c}=0. \label{consec}
\feqn
The first identity gives the very interesting relations
\eqn
\alpha_{Ia}^{\ \ c}=(C_I)^c_{\ a}, \label{relaz}
\feqn
where $\alpha_{Ia}^{\ \ c}$ are some of the structure constants of
$E_6$, the ones defined by $[C_I, \tilde C_a]^d_{\ b}=\alpha_{Ia}^{\ \ c} (\tilde C_c)^d_{\ b}$.
Note that (\ref{relaz}) relates the structure constants of $E_6$ (directly computable from the adjoint representation) to the matrices
of the smallest fundamental representation of $F_4$.\\
Another interesting information comes from identity (\ref{consec}). With our normalization, ${\rm Tr}(j_a\star j_b)=0$
implies
\eqn
(\tilde C_a)^{27}_{\ \ b}=-i\ \sqrt{\frac \tau3} \delta_{ab} \label{27}.
\feqn
Inserted in (\ref{consec}) it gives
\eqn
(C_I)^c_{\ a} \delta_{cb}+(C_I)^c_{\ b} \delta_{ca}=0.
\feqn
Then, the matrices $(C_I)^c_{\ a}$ are antisymmetric!

%%%%%%%%%%%%%%%%%%%%%%%%%%%%%%%%%%%%%%%%%%%%%%%%%%%%%%%%%%%%%%%%%%%%%%%%%%%%%%%%
\subsection{The Yokota construction and {\bf 56}}\label{sec:tits2}
We can obtain the representation ${\bf 56}$ by applying the method explained by Yokota in \cite{IY}. We will first consider
the general complex realization and then we will specialize to the compact form. From our previous analysis we know that the
Lie algebra can be written as
\begin{eqnarray}
  E_7=E_6 \oplus i\otimes \mathbb{R} \oplus j \otimes \jj \oplus k\otimes \jj,
\end{eqnarray}
where we have included the remaining derivative generators $D_j$ and $D_k$ in the last two addends. As supporting space we
take the $56$-dimensional space
\begin{eqnarray}
  V_{56}=(\jj \oplus \mathbb{C})^2.
\end{eqnarray}
One has to define an action of $E_7$ on this space. To this end it is convenient to introduce some further geometry.
\subsection*{The determinant form.}
The Jordan algebra is endowed with a trilinear form, the determinant form
\begin{eqnarray}
  Det: \jj \times \jj \times \jj \longrightarrow \mathbb{C},
\end{eqnarray}
defined by
\begin{eqnarray}
  Det(j_1,j_2,j_3)=\frac 13 \tr(j_1 \circ j_2 \circ j_3)-\frac 16 \left( \tr(j_1) \tr(j_2 \circ j_3) + \tr(j_2) \tr(j_1 \circ j_3)+
  \tr(j_3) \tr(j_1 \circ j_2) \right) + \frac 16 \tr(j_1) \tr(j_2) \tr(j_3).
\label{determinant}
\end{eqnarray}
This is a fundamental ingredient in realizing exceptional Lie algebras. For example, it is left invariant by the action of $E_6$ on $\jj$. Indeed,
the group $E_6$ can be defined as the group of endomorphisms of $\jj$ that leave the determinant form invariant. The form $Det$ is a completely
symmetric tensor $D$, also called the cubic invariant of $E_6$, with component $D_{\alpha\beta\gamma}=Det(j_\alpha, j_\beta, j_\gamma)$.
$Det$ is non degenerate. This means that it induces an action $\rhd$ of $\jj$ on itself
\begin{eqnarray}
  \rhd : \jj \times \jj & \longrightarrow & \jj, \\
          (j_1,j_2) & \longmapsto & j_1 \rhd j_2,
\end{eqnarray}
defined by the relation $Det(j_1,j_2,j_3)=:\frac 13 \tr\big( (j_1 \rhd j_2) \circ j_3 \big)$.
This is called the Freudenthal product. More explicitly:
\begin{eqnarray}
j_1 \rhd j_2=j_1 \circ j_2-\frac 12 \tr(j_1) j_2 - \frac 12 \tr(j_2) j_1 +\frac 12 \tr(j_1) \tr(j_2) I_3 -\frac 12 \tr(j_1 \circ j_2)  I_3.
\label{jaction}
\end{eqnarray}

\

\

With this richer structure at hand we can define the action of $E_7$ on $V_{56}$. Given $X\in E_6$, let $\phi_X$ its image under the fundamental
representation ${\bf 27}$ (for example constructed in \cite{E6}). Then, the image of $X$ under ${\bf 27'}$ (the second $27$-dimensional representation
of $E_6$) is $-\phi_X^t$. For $g=(X, \nu, j_1, j_2)\equiv X+i\otimes \nu +j\otimes j_1 +k\otimes j_2 \in E_7$ and
$v:=(\tilde j_1, \mu_1, \tilde j_2, \mu_2)^t \in V_{56}$, we define \cite{IY}
\begin{eqnarray}
  gv:=\begin{pmatrix}
    \phi_X \tilde j_1 +\frac i2 \sqrt{\frac \tau3} \nu \tilde j_1 +(j_2 +i\otimes j_1) \rhd \tilde j_2 +\frac 12 \mu_2 (-j_2+i\otimes j_1) \\
    \frac 12 \langle (-j_2+i\otimes j_1), \tilde j_2\rangle -\frac i2 \sqrt {3\tau} \nu \mu_1 \\
    -\phi^t_X \tilde j_2 -\frac i2 \sqrt{\frac \tau3} \nu \tilde j_2 +(-j_2 +i\otimes j_1) \rhd \tilde j_1 +\frac 12 \mu_1 (-j_1+i\otimes j_1) \\
    \frac 12 \langle (j_2+i\otimes j_1), \tilde j_1\rangle +\frac i2 \sqrt {3\tau} \nu \mu_2
  \end{pmatrix}.
\end{eqnarray}
Note that w.r.t. \cite{IY} we have fixed a prescription $2B=j_2+i\otimes j_1$, $2A=-j_2+i\otimes j_1$ and $\nu \mapsto -\frac i2 \sqrt{\frac 32} \nu$
in order to get a compact real form. $\tau$ is the normalization in the trace product of the Jordan basis:
$\tr (j_\alpha \circ j_\beta)=\tau \delta_{\alpha\beta}$. The remaining coefficient are chosen so that the corresponding matrices have all the same
normalization w.r.t. the trace product. An explicit realization with $\tau=2$ is given in appendix \ref{sec:56}.\\
It is wort to look at the general form of the resulting matrices in the ${\bf 56}$. Let us choose the normalization $\tau=2$:
\eqn
&& Y=\left(
\begin{array} {c|c|c|c}
\phi +i\frac \nu{\sqrt 6}  & \ov & -\sum_{\alpha=1}^{27} \bar z^\alpha A_\alpha & \frac 1{\sqrt 2} \sum_{\alpha=1}^{27} z^\alpha \vec e_\alpha \cr
\hline
\ov^t  & -i\nu \sqrt {\frac 32} & \frac 1{\sqrt 2} \sum_{\alpha=1}^{27} z^\alpha \vec e_\alpha^t & 0 \cr
\hline
\sum_{\alpha=1}^{27} z^\alpha A_\alpha  & -\frac 1{\sqrt 2} \sum_{\alpha=1}^{27} \bar z^\alpha \vec e_\alpha & -\phi^t -i\frac \nu{\sqrt 6} & \ov \cr
\hline
-\frac 1{\sqrt 2} \sum_{\alpha=1}^{27} \bar z^\alpha \vec e_\alpha^t  & 0 & \ov^t & i\nu \sqrt {\frac 32}
\end{array}
\right),
\feqn
where $\phi$ is in the ${\bf 27}$ of $E_6$, $\nu\in \mathbb{R}$ and $z^\alpha \in \mathbb{C}$. The matrices $A_\alpha$ have components
\begin{eqnarray}
  (A_\alpha)_{\beta\gamma}=\frac 12 \tr ((j_\alpha \rhd j_\gamma)\circ j_\beta)=\frac 32 D_{\alpha\gamma\beta}.
\end{eqnarray}
Thus, the matrices $A_\alpha$, $\alpha=1,\ldots, 27$ are determined by the invariant cubic tensor $D$ and are then symmetric and totally symmetric
by including the index $\alpha$. Actually we can say more. Assume $a,b,c =1,\ldots, 26$. Then $\tr j_a =\tr j_b =\tr j_c=0$ and
\begin{eqnarray}
  D_{abc}=(A_c)_{ab}=\frac 12 \tr (j_c \circ j_b \circ j_a) \equiv i (\tilde C_c)^a_{\ b}.
\end{eqnarray}
Then, also the matrices $\tilde C_a$, when we drop the 27-th row and column, are totally symmetric. This implies (\ref{symm}) as a particular case.\\
By applying (\ref{jaction}) to compute the remaining components, we get
\begin{eqnarray}
  && 0=\frac 32 D_{a,27,27}=(A_a)_{27,27}=(A_{27})_{a,27}=  (\tilde C_a)^{27}_{\ \ 27}=(\tilde C_{27})^{a}_{\ 27}; \\
  && -\frac 12 \sqrt {\frac 23} \delta_{ab}= \frac 32 D_{a,b,27}=(A_a)_{b,27}=(A_{27})_{ab} =-\frac i2  (\tilde C_a)^b_{\ 27}=-\frac i2
  (\tilde C_{27})^b_{\ a};\\
  && \sqrt {\frac 23} =\frac 32 D_{27,27,27}=(A_{27})_{27,27}=i (\tilde C_{27})^{27}_{\ \ 27}.
\end{eqnarray}
These interesting relations provide a simple way to construct the matrices $\tilde C_a$ which extend the $F_4$ algebra to the $E_6$ algebra.

%%%%%%%%%%%%%%%%%%%%%%%%%%%%%%%%%%%%%%%%%%%%%%%%%%%%%%%%%%%%%%%%%%%%%%%%%%%%%%%
\subsection{Construction of the group}\label{sec:group}
From the Tits construction we can  easily get the main features required to perform the Euler construction with respect
to the maximal subgroup $(E_6\times U(1))/{\mathbb {Z}_3}$.
If $\vec n$ is any normalized three vector, we see that setting $\vec h =(h_1,h_2,h_3)$ and $\vec H =(H_1,H_2,H_3)$,
\eqn
\mathfrak{u}:= \vec n \cdot \vec H \oplus \jj \dot + ((\vec n \cdot \vec h) \otimes \jj')
\feqn
is a subalgebra of $\g$ which generates a maximal compact subgroup $U(1)\times E_6$. Indeed, note that $\vec n \cdot \vec H$ has vanishing
Lie product with the whole $\mathfrak {u}$. The linear complement of $\vec n \cdot \vec H$ in $\mathfrak {u}$ reproduces the same rules we
used in \cite{E6} to extend $F_4$ to $E_6$. In particular, restricted to $(\vec n \cdot \vec h) \otimes \jj'$, the product (\ref{prod})
is $[(\vec n \cdot \vec h)\otimes j, \vec n \cdot \vec h\otimes j']:= -[L_j, L_{j'}]$. The minus sign is exactly what we needed in
\cite{E6} to go from the split form $E_{6(-26)}$ to the compact form. This is not surprising since we are starting from a compact group,
but it can be considered as a consistency check. Note that we thus have a family of $E_6\times U(1)$ subgroups parameterized by
the choice of the vector $\vec n$. On this family there is the action of an $SU(2)$ subgroup changing $\vec n$, underlying the presence
of a quaternionic structure. However, note also that in order to realize the quotient space $E_7/((E_6\times U(1))/{\mathbb {Z}_3})$
one has to fix the $\vec n$ thus breaking the $SU(2)$ structure and consequently the quaternionic structure.\\
For definiteness, we will choose $\vec n=e_1$, where $e_i$, $i=1,2,3$ is the canonical basis for $\mathbb {R}^3$. Thus
\eqn
\mathfrak{u}=\lie (U(1)\times E_6)=H_1 \oplus D(\jj) \dot + (h_1 \otimes \jj').
\feqn
Having selected the maximal subgroup $U=\exp \h$, we look at the construction of the Euler parametrization \cite{F4}
\eqn
E_7=B e^V U.
\feqn
Here $V$ is a maximal subspace of the linear complement $\p$ of $\h$ in $\g$ such that $Ad_U (V)=\p$, whereas $B=U/U_o$, where
$U_o$ is the kernel of the map:
\bes
{\rm Ad}^o:U \rightarrow {\rm Aut}({\rm Lie}(U)), \qquad u\mapsto {\rm
  Ad}^o(U):={\rm Ad}_u|_V,
\ees
the apex $o$ means restriction to $V$. Let us begin with the selection of $V$. We take $H_2$ as
a first generator of $V$. Acting on it with $ad_\h$ we generate $H_3$ and $h_3\otimes \jj'$. To generate $h_2\otimes \jj'$ we need
to add at least an element of the form $h_2\otimes j_a$. Now, $\jj'$ contains a particular two dimensional subspace $W$ which corresponds
to the vector space of diagonal traceless Jordan matrices. Let $\{j_1, j_{18}\}$ be a basis for this subspace. Up to now, the indices
we have chosen are arbitrary labels. However, we are now referring to \cite{E6} where we have chosen a well defined basis for $\jj'$. In that basis
the diagonal generators are just $j_1$ and $j_{18}$ which correspond to the matrices $c_{53}$ and $c_{70}$ of $E_6$.
Note that the elements of $h_2\otimes W$ commute w.r.t. the product (\ref{prod}). If we fix, for example, the element
$h_2\otimes j_1$, acting on it with $ad_{D(\jj)}$ we will generate all the basis elements $h_2\otimes j_a$ but $h_2\otimes j_{18}$. We conclude that
\eqn
V=\RR H_2 \oplus W=\RR H_2 \oplus \RR (h_2\otimes j_1)\oplus \RR (h_2\otimes j_{18}).
\feqn
At this point we have that $\dd U_o =2\dd \h +\dd V-\dd \g=28$. On the other hand, we already know what $U_o$ is, it has, in fact,  been studied in \cite{E6}.
Indeed, if $\h_o=\lie (U_o)$, then $\h_o$ is the subset of $\h$ of elements which commute with $H_2, h_2\otimes j_1$ and $h_2\otimes j_{18}$.
But this is the subset of $D(\jj)$, commuting with $h_2\otimes j_1$ and $h_2\otimes j_{18}$, which determine exactly the subgroup $SO(8)$
of $F_4$ commuting with $W$, studied in \cite{E6} (in that case $W$ was generated by $c_{53}$ and $c_{70}$).\\
This allows us to provide the final expression for the general element of the group. Indeed, if $\Psi_A$, $A=1,\ldots,133$ is the basis of $\g$
defined above, and $c_s$, $s=1,\ldots, 78$ is the matrix representation of $\lie(E_6)$ given in
\cite{E6}, then the map
\eqn
\psi: \lie(E_6) \hookrightarrow \g, \quad c_s \mapsto \Psi_{s+3}
\feqn
gives an embedding of $E_6$ in $E_7$. This provides the generic element
$$
U(x_1, \ldots, x_{79})=\exp(x_1 \Psi_1) \psi_*(E_6 [x_2,\ldots,x_{79}]),
$$
where $E_6$ is the parametrization given in \cite{E6} and $\psi_*$ is
the push forward of $\psi$ under the exponential map. Next, $B$ is easily obtained from $E_6$ by dropping the last $28$ factors on
the right, obtaining
$$
B[x_1,\ldots, x_{51}]=U(x_1,\ldots,x_{51},0,\ldots,0).
$$
The basis of $V$ being given by $\{\Psi_2, \Psi_{82}, \Psi_{99}\}$, we finally get
\eqn
E_7[x_1,\ldots,x_{133}]= e^{x_{1} \Psi_1} \psi_*(E_6 [x_2,\ldots,x_{51},0,\ldots,0]) e^{x_{52} \Psi_2+x_{53} \Psi_{82}+x_{54} \Psi_{99}}
e^{x_{55} \Psi_1} \psi_*(E_6 [x_{56},\ldots,x_{133}]). \label{gruppoE7}
\feqn
Thus we are left with the problem of determining the range of the parameters. This can be done by means of the topological method developed
in \cite{F4}. Concretely, this consist in choosing the range of parameters in order to define a $133$ dimensional closed cycle. This will
eventually cover the group an integer number $N$ of times, so that one must finally reduce the range to have $N=1$.
%Concretely, this consists in choosing range of the variables so that the last ones cover the whole maximal
%subgroup and the remaining ones complete the given subset to a closed cycle which covers the group one time.
%(*qui e' poco chiaro*)
%Before doing that some further comments are in order.
The explicit deduction of the range of parameters is done in Appendix \ref{sec:parameters}. We conclude this section with some further comments.

%%%%%%%%%%%%%%%%%%%%%%%%%%%%%%%%%%%%%%%%%%%%%%%%%%%%%%%%%%%%%%%%%%%%%%%%%%%%%%%%%%%%%%
\subsubsection{Remarks}
To realize concretely the group we could use either the adjoint representation or the $\bf {56}$. However, the kernel of
the Adjoint representation of the group is its center, which
is $\mathbb{Z}_2$ for $E_7$, so that using it in (\ref{gruppoE7}) will provide the group $E_7/\mathbb{Z}_2$ in place of $E_7$. Indeed, $\ker Ad$ is, by construction, the subset of $G$
which commute with all $G$.
On the contrary, the representation $\bf {56}_G$ of the group (we add a suffix $G$ to distinguish the group from the algebra)
is faithful so that we can use $\bf {56}$ to construct $E_7$.
This is a well known fact, but we can check it directly from our construction. Let $M_A$ be the basis of
the algebra $E_7$ in the adjoint representation, and $Y_A$ the corresponding basis in the $\bf{56}$.
They have exactly the same
structure constants, so that they can be thought of as  representing the same elements of the algebra.
However, $\bf {56}$ contains the nontrivial
generator of the center
\begin{equation}
-I_{56}=\exp (\sqrt 6 \pi Y_1).
\end{equation}
Obviously the Adjoint representation cannot contain $-I_{133}$ and, indeed, $\exp {\sqrt 6 \pi M_1}=+I_{133}$. Thus, the correspondence
\begin{eqnarray}
\xi: {\bf{56}} \longrightarrow {\bf{133}}, \quad Y_A \longmapsto M_A, \qquad\ A=1,\ldots,133,
\end{eqnarray}
defines a surjective homomorphism
\begin{eqnarray}
\Xi:{\bf{56}_G} \longrightarrow {\bf{133}_G}, \quad \exp {\sum_{A=1}^{133} {\lambda^A Y_A}} \longmapsto \exp {\sum_{A=1}^{133} {\lambda^A M_A}},
\end{eqnarray}
which has kernel $\ker \Xi=\mathbb{Z}_2=\{I_{56}, -I_{56}\}$.
Then, $\bf{56}_G$ is a double covering of $\bf{133}_G$.
\\
In particular, let us consider the corresponding one parameter subgroups:
\begin{eqnarray}
&& h_A(t)=\exp (t M_A), \qquad\ A=1,\ldots 133, \qquad\ t\in \mathbb{R},\\
&& g_A(t)=\exp (t Y_A), \qquad\ A=1,\ldots 133, \qquad\ t\in \mathbb{R}.
\end{eqnarray}
We then note that for $A>3$ all of them have period $4\pi$, apart from $A=73, 99$ and $125$, which have period $4\pi\sqrt 3$,
whereas for $A=1,2,3$, $h_A$ has period $T_h=\sqrt {6} \pi$ and $g_A$
has period $T_g=\sqrt 6\ 2\pi=2T_h$ as a consequence of the double covering.\\
Finally, we remark that the maximal subgroup of $E_7$ is $(E_6\times U(1))/\mathbb{Z}_3$.
Indeed, the generator
\begin{eqnarray}
\omega:=\exp (4\frac \pi{\sqrt 3} Y_{73})= \exp (2\sqrt{\frac 23} \pi Y_1)\in E_6\cap U(1)
\end{eqnarray}
satisfy $\omega^3=I_{56}$ so that it is a generator of $\mathbb{Z}_3$ common to both $E_6$ and $U(1)$.

%%%%%%%%%%%%%%%%%%%%%%%%%%%%%%%%%%%%%%%%%%%%%%%%%%%%%%%%%%%%%%%%%%%%%%%%%%%%%%%%%%%%%%%%%%%%%%%%%%%%%%%%%%%%%%%%%%%%%%%%%%%%%
%%%%%%%%%%%%%%%%%%%%%%%%%%%%%%%%%%%%%%%%%%%%%%%%%%%%%%%%%%%%%%%%%%%%%%%%%%%%%%%%%%%%%%%%%%%%%%%%%%%%%%%%%%%%%%%%%%%%%%%%%%%%%
\section{The $E_{7(7)}$ construction}\label{sec:E7(7)}
%\section{The Adams construction}\label{adamsalg}
We want to construct the $E_7$ compact form related to the split form by means of the Weyl unitary trick. In this case the maximal compact subgroup that is
privileged is $SU(8)/\mathbb{Z}_2$. To this aim we will follow the paper \cite{adams}, chapter 12, see also \cite{Cacciatori:2010ws}.\\
Let $V$ be an eight dimensional real vector space and $V^*$ be its dual.
Let $\wedge^iV$ be the $i$-th external power of $V$. We can fix an isomorphism $\wedge^8V\simeq \rr$. $\SL(V)$ is the group
of automorphisms preserving such isomorphism. Let $L:=\sla(V)$ be its Lie algebra. We will construct the representation ${\bf 56}$ of $E_7$ by
extending the representation of $L$ on $W:=\wedge^2V\oplus\wedge^2V^*$ to a representation of $E_7$.
Set $A:=L\otimes\wedge^4 V$. We wish to see $A$ acting as a Lie algebra of linear maps $W\rightarrow W$.
The action of $L$ on $W$ is as usual:
\be
L(W)=L(V)\wedge V\oplus L(V^*)\wedge V^*+V\wedge L(V)\oplus V^*\wedge L(V^*),
\ee
where $L(V^*)$ is the adjoint action.
If $i+j=8$ the pairing $\wedge^i V\otimes\wedge^j V\rightarrow \wedge^8V\simeq\rr$, given by the wedge product, defines an isomorphism
$\wedge^iV\simeq\wedge^jV^*$. This isomorphism can be used to define the second component of the action of $A$.
Actually, an action of $\lambda^4\equiv\wedge^4V$ on $W$ can be obtained as follows:
\begin{align}\label{id}
&\lambda^4\otimes\wedge^2V\stackrel{\wedge}{\longrightarrow}\wedge^6V\simeq\wedge^2V^* \cr
&\lambda^4\otimes\wedge^2V^*\simeq\wedge^4V^*\otimes\wedge^2V^*\stackrel{\wedge}{\longrightarrow}\wedge^6V^*\simeq\wedge^2V,
\end{align}
where $\wedge$ is the usual multiplication in the exterior algebra $\wedge V$. Thus, $A$ is a $133$-dimensional real vector space of operators acting on $W$,
that indeed realizes an $E_7$ Lie algebra representation, more precisely the $E_{7(7)}$ split form. In order to see this and realize the compact
form, let us look more carefully at the explicit matrix realization.
%%%%%%%%%%%%%%%%%%%%%%%%%%%%%%%%%%%%%%%%%%%%%%%%%%%%%%%%%%%%%%%%%%%%%%%%%%%%%%%%%%%%%%%%%%%%%%%%%%%%%%%%%%%%%%%%

\subsection{Matrix realization}
We identify $V$ with $\rr^8$. The action of $L$ on $V$ is generated by the action of all traceless matrices in $M(8,\rr)$. Fix a basis $\{e_i\}_{i=1}^8$ of $V$
and define a basis $\{A_{kl},S_{kl},D_\alpha\}$ for $M(8,\rr)$, where $1\leq k<l\leq 8$, $\alpha=1,\cdots, 7$ and
$A_{kl}e_i=\delta_{li}e_k-\delta_{ki}e_l$, $S_{kl}e_i=\delta_{li}e_k+\delta_{ki}e_l$, and $D_\alpha$ form a basis of diagonal traceless matrices
$D_\alpha=\diag(D^1_\alpha,\cdots,D_\alpha^8)$.
We can normalize them as $\trace(D_\alpha D_\beta)=2\delta_{\alpha\beta}$, so that all matrices are orthogonal, the symmetric matrices are normalized to $2$ and
the antisymmetric ones to $-2$ w.r.t. the trace product.\\
In order to write the action of this basis on $W$ let us introduce the following notations:
\begin{itemize}
\item we select a basis $e_{ij}:=e_i \wedge e_j$, $i<j$ of $V\wedge V$ and the canonical dual basis $\varepsilon^{ij}$. As usual we will extend the range
of the indices $i,j$ as running independently from $1$ to $8$, by assuming antisymmetry (so $e_{ij}=-e_{ji}$ and so on).
\item A vector $v\in V\wedge V$ can be then written as
$$
v=\frac 12 \sum_{i,j} v^{ij} e_{ij} =\sum_{i<j} v^{ij} e_{ij},
$$
and a vector $w\in V^*\wedge V^*$ as
$$
w=\frac 12 \sum_{i,j} w_{ij} \varepsilon^{ij} =\sum_{i<j} w_{ij} \varepsilon^{ij}.
$$
\item A linear operator $M:V\wedge V\rightarrow V\wedge V$ acts on the components as
$$
(Mv)^{ij}=\sum_{i<j}M^{ij}_{\ \ kl}v^{kl},
$$
and similar notations for the other possibilities $[V^*\wedge V^*\rightarrow V^*\wedge V^*]$,
$[V^*\wedge V^*\rightarrow V\wedge V]$ and $[V\wedge V\rightarrow V^*\wedge V^*]$.
\end{itemize}
We can then easily write down the explicit matrix action of $L$ over $\rr^{56}\simeq V\wedge V\oplus V^*\wedge V^*$:
\begin{eqnarray}
  && A_{kl}= \begin{pmatrix}
    (A_{kl}^u)^{ij}_{\ \ i'j'} & 0 \\ 0 & (A_{kl}^d)_{ij}^{\ \ i'j'}
  \end{pmatrix} \cr
  && \phantom{A_{kl}}= \begin{pmatrix}
   (\delta_{ki'}\delta_{li}-\delta_{ki}\delta_{li'})\delta_{jj'}+
   (\delta_{kj'}\delta_{lj}-\delta_{kj}\delta_{lj'})\delta_{ii'} & 0 \\ 0 & (\delta_{ki'}\delta_{li}-\delta_{ki}\delta_{li'})\delta_{jj'}+
   (\delta_{kj'}\delta_{lj}-\delta_{kj}\delta_{lj'})\delta_{ii'}
  \end{pmatrix}\\
  && S_{kl}= \begin{pmatrix}
    (S_{kl}^u)^{ij}_{\ \ i'j'} & 0 \\ 0 & (S_{kl}^d)_{ij}^{\ \ i'j'}
  \end{pmatrix} \cr
  && \phantom{A_{kl}}= \begin{pmatrix}
   i(\delta_{ki'}\delta_{li}+\delta_{ki}\delta_{li'})\delta_{jj'}+
   i(\delta_{kj'}\delta_{lj}+\delta_{kj}\delta_{lj'})\delta_{ii'} & 0 \\ 0 & -i(\delta_{ki'}\delta_{li}+\delta_{ki}\delta_{li'})\delta_{jj'}-
   i(\delta_{kj'}\delta_{lj}+\delta_{kj}\delta_{lj'})\delta_{ii'}
  \end{pmatrix} \\
  && D_{\alpha}= \begin{pmatrix}
    (D_{\alpha}^u)^{ij}_{\ \ i'j'} & 0 \\ 0 & (D_{\alpha}^d)_{ij}^{\ \ i'j'}
  \end{pmatrix} = \begin{pmatrix}
    i(D_\alpha^i+D_\alpha^j) \delta^{ij}_{i'j'} & 0 \\ 0 & -i(D_\alpha^i+D_\alpha^j) \delta^{i'j'}_{ij}
  \end{pmatrix},
\end{eqnarray}
where $\delta^{ij}_{i'j'}$ is the identity over $V\wedge V$. Note that, in order to obtain the compact form, we have multiplied by $i$ the symmetric matrices.\\
For the remaining 70 generators we have to consider the action of $\{\lambda_{i_1i_2i_3i_4}=e_{i_1}\wedge e_{i_2}\wedge e_{i_3}\wedge e_{i_4}\}_{i_1<i_2<i_3<i_4}$
on $W$. This is easily realized by implementing the identifications \eqref{id}:
\begin{align}
(\lambda_{i_1i_2i_3i_4})\otimes (e_{j_1j_2})&\mapsto \frac 12 \epsilon_{i_1i_2i_3i_4j_1j_2k_1k_2}\varepsilon^{k_1k_2} \cr
(\lambda_{i_1i_2i_3i_4})\otimes (\varepsilon_{j_1j_2})&\mapsto \frac 12 \delta^{j_1j_2k_1k_2}_{i_1i_2i_3i_4}e_{k_1k_2},
\end{align}
where $\epsilon$ is the standard eight dimensional Levi-Civita tensor and
\be
\delta^{j_1j_2j_3j_4}_{i_1i_2i_3i_4}=\sum_{\sigma\in\mathcal{P}}\epsilon_\sigma\delta_{i_{\sigma(1)}}^{j_1}\delta_{i_{\sigma(2)}}^{j_2}
\delta_{i_{\sigma(3)}}^{j_3}\delta_{i_{\sigma(4)}}^{j_4}
\ee
with $\mathcal{P}$ the set of permutations and $\epsilon_\sigma$ is the parity of $\sigma$.
The action of $\lambda_{i_1i_2i_3i_4}$ in the block matrix form, with respect to the decomposition $W=\wedge^2V\oplus \wedge^2V^*$, is then
\be
\lambda_{i_1i_2i_3i_4}=
\begin{pmatrix}
0 & (\lambda^u_{i_1i_2i_3i_4})_{ij,kl} \\
(\lambda^d_{i_1i_2i_3i_4})^{ij,kl} & 0
\end{pmatrix}
=
\begin{pmatrix}
0 & \epsilon_{i_1i_2i_3i_4ijkl} \\
\delta^{ijkl}_{i_1i_2i_3i_4} & 0
\end{pmatrix}.
\ee
Note that the matrices $\lambda^u$ and $\lambda^d$ are both symmetric. Let us introduce the ordered tetra-indices $I\equiv \{i_1i_2i_3i_4\}$ with the
rule $i_1<i_2<i_3<i_4$. Its complementary is the unique ordered tetra-index $\tilde{I}$ such that $\epsilon_{I\tilde{I}}\neq 0$. Then
\be
\transp{\lambda}_I=\epsilon_{I\tilde{I}}\lambda_{\tilde{I}}.
\ee
Thus, we can change basis for $\lambda^4$ introducing symmetric matrices
\be
\mathcal{S}_I:=\frac{i}{\sqrt{2}}(\lambda_I+\epsilon_{I\tilde{I}}\lambda_{\tilde{I}})
\ee
and antisymmetric matrices
\be
\mathcal{A}_I:=\frac{1}{\sqrt{2}}(\lambda_I-\epsilon_{I\tilde{I}}\lambda_{\tilde{I}}).
\ee
Again, we have included the imaginary unit for the symmetric matrices.
The cardinality of the set of tetra-indices is 70 so that only half of the $\mathcal{S}_I$ (and of the $\mathcal{A}_I$) can be linearly independent.
Indeed, we have $\mathcal{S}_I=\mathcal{S}_{\tilde{I}}$ ($\mathcal{A}_I=-\mathcal{A}_{\tilde{I}}$).
To avoid this double over-counting we can restrict ourselves to the subset $\mathcal{I}_0$ of tetra-indices defined in appendix \ref{app:range77}.
Thus, a basis for $\wedge^4 V$ (as linear operators over $W$) is:
\be
\{\mathcal{S}_I,\mathcal{A}_I\}_{I\in\mathcal{I}_0}.
\ee
A basis for $A\equiv E_7$ is then
\be
\{A_{kl},\mathcal{A}_I,D_\alpha,S_{kl},\mathcal{S}_I \}_{1\leq k<l\leq 8;\,1\leq\alpha\leq 8;\,I\in\mathcal{I}_0}.
\ee
All matrices are orthogonal. The antisymmetric matrices $A_\mu\equiv\{A_{kl},\mathcal{A}_I\}$ are normalized by $\trace (A_\mu A_\nu)=-2\delta_{\mu\nu}$
and have cardinality $28+35=63$ so that generate the maximal compact subgroup $\SU(8)/\mathbb{Z}_2$. The remaining 70 symmetric generators
$S_\Lambda\equiv\{D_\alpha,S_{kl},S_I\}$ are normalized by $\trace S\lambda S_M=2\delta_{\Lambda M}$, so that are the noncompact part of the algebra.
In particular, the 7 diagonal matrices $D_\alpha$ generate a Cartan subalgebra.

%%%%%%%%%%%%%%%%%%%%%%%%%%%%%%%%%%%%%%%%%%%%%%%%%%%%%%%%%%%%%%%%%%%%%%%%
\subsection*{Summarizing.}
The above construction furnishes the Lie algebra of $E_{7(-133)}$, but the split form can be recovered simply by dropping the imaginary unit $i$
from the symmetric matrices. \\
Let us summarize the structure of the matrices. 63 of them are block diagonal
\be
M_i=
\begin{pmatrix}
\sla(8) & 0 \\
0 & \sla(8)
\end{pmatrix},
\qquad i=1,\cdots,63,
\ee
where the two diagonal block are the 28-dimensional representation of $\sla(8)$ on $\wedge^2 V$ and $\wedge^2 V^*$ respectively.
A basis for $\sla(8)$ is composed by 35 symmetric matrices and 28 antisymmetric matrices. $\SL(8)$ contains the maximal compact subgroup $\SO(8)$.
The 28 antisymmetric matrices generate its Lie algebra in the compact form. Of the other 35 symmetric matrices 7 are diagonal with vanishing trace and 28 are
symmetric with all diagonal elements equal to zero. The 7 diagonal elements generate a Cartan subalgebra of $\sla(8)$ and, obviously, also of $E_7$.
The remaining 70 matrices have the structure
\be
M_i=
\begin{pmatrix}
0 & \wedge^4 V \\
\wedge^4 V & 0
\end{pmatrix}, \qquad i=70,\cdots, 133.
\ee
Among these matrices 35 are symmetric and 35 antisymmetric. The group $\SU(8)/\mathbb{Z}_2$ is a maximal compact subgroup of $E_7$ and its Lie algebra is
generated by the 63 antisymmetric matrices. Thus we can write $e_7=\su(8)\oplus \mathfrak{p}$, where $\mathfrak{p}$ is the complement of $\su(8)$ in $e_7$.

%%%%%%%%%%%%%%%%%%%%%%%%%%%%%%%%%%%%%%%%%%%%%%%%%%%%%%%%%%%%%%%%%%%%%%%%%%%%%%%%%%%%%%%%%%%%%%%%%%%%%%%%%%%%%%%%%%%%%%%%%%%%%%%%%%%%%%%%%%%%%%%%%%%%%
\subsection{Construction of the group}
The strategy for constructing the generalized Euler parametrization of the group is the same as for the previous construction so that we will only
sketch the main steps. In this case the reference subgroup is the smallest maximal compact subgroup $\SU(8)/\zz_2$.
Let $\mathfrak{u}=\lie(\SU(8))$ and, as before, $\mathfrak{g}=\lie(E_7)\equiv e_7$ and $\mathfrak{p}$ the complement of $\mathfrak{u}$ in $\mathfrak{g}$.
Having selected the maximal subgroup $U=\exp \mathfrak{u}$ we look at the construction of the Euler parametrization \cite{F4}:
\be \label{eul}
E_7=Be^VU.
\ee
Here $V$ is a maximal subspace of the linear complement $\mathfrak{p}$ of $\mathfrak{u}$ in $\mathfrak{g}$ such that $\Ad_U(V)=\mathfrak{p}$,
whereas $B=U/U_0$, where $U_0$ is the kernel of the map
\be
\Ad^0:U\rightarrow\Aut(\lie(U)),\qquad u\mapsto \Ad^0(U):=\Ad_u|_V,
\ee
the apex 0 means restriction to $V$.
As $V$ we can choose the 7 diagonal matrices of a Cartan subalgebra of $e_7$: they form an Abelian subalgebra that can be used to generate
$\mathfrak{p}$ by means of the adjoint action of $U$ (a general proof of these statements will appear in \cite{to-appear}).
For dimension reasons it follows that $U_0$, the kernel of $\Ad^0$, is at most a discrete subgroup. Indeed, it can be shown that $U_0=\mathbb{Z}_2^7$,
\cite{to-appear}.\\
Thus, we can introduce coordinates $\underline x =(x_1,\ldots, x_{63})$, $\underline z =(z_1,\ldots, z_{63})$ and
$\vec y =(y_1,\ldots, y_{7})$ so that\footnote{Concretely, $B[\underline x]=U[\underline x]$ and the difference is only in the range.}
$E_7=U[\underline{x}] e^{V[\vec y ]} U[\underline{z}]$, where $U[\underline{z}]$ is a parametrization
of $\SU(8)/\zz_2$, the range of the parameters $\underline x$ is reduced by the action of $U_0$, and
$$
V[\vec y ]=\sum_{\alpha=1}^7 y^\alpha D_\alpha.
$$
We will now focus on the determination of the range for $\vec y$. The ranges for the coordinates $\underline x$ and $\underline z$, can be easily determined,
for example, as in \cite{Bertini:2005rc}.\\
Let $\mathfrak t$ be the complement of $V$ in $\mathfrak p$.
The general strategy developed in \cite{F4} shows that the invariant measure over $E_7$ is
\begin{eqnarray}
&& d\mu_{E_7}=d\mu_{U} [\underline x]\ d\mu_{U} [\underline z]\ |f(\vec y)|\ d\vec y^7, \\
&& f(\vec y ):={\rm det} [\Pi \circ {\rm Ad}_{e^{-V}} : \mathfrak{u} \rightarrow \mathfrak {t} ],
\end{eqnarray}
where $\Pi$ is the orthogonal projection on $\mathfrak t$. We assume that one always works with orthonormal bases, so that the determinant function
is well defined up to an irrelevant sign (we need only the modulus). The function $f$ is determined in appendix \ref{app:range77}, the final expression is
\begin{eqnarray}
  |f(\vec y)|=\prod_{\beta \in {\rm Rad}^+} \sin (|\sum_{a=1}^7 y^a \beta(D_a)| ),
\end{eqnarray}
where ${\rm Rad}^+$ is the set of positive roots w.r.t. $V$.
Thus, the equations for the range of $\vec y$ are
\begin{eqnarray}
0< |\sum_{a=1}^7 y^a \beta(D_a)|<\pi, \qquad\ \beta \in {\rm Rad}^+.
\end{eqnarray}
This is a set of $63$ double inequalities, which, however, can be quickly reduced to a set of eight equations as follows. Indeed, all positive
roots can be obtained as non negative integer linear combinations of the simple roots. In particular, there exists a unique longest root
$$\beta_{max}=2\alpha_1+2\alpha_2+3\alpha_3+4\alpha_4+3\alpha_5+2\alpha_6+\alpha_7$$
whose coefficients are the highest ones. From this it follows that all inequalities are a consequence of the ones corresponding to the simple roots plus the
one associated to the longest one. These are
\begin{eqnarray}
&& 0< \frac 12 (y^1-y^2-y^3-y^4-y^5-y^6+\sqrt 2 y^7) <\pi, \\
&& 0< y^1+y^2 <\pi, \\
&& 0< -y^1+y^2 <\pi, \\
&& 0< -y^2+y^3 <\pi, \\
&& 0< -y^3+y^4 <\pi, \\
&& 0< -y^4+y^5 <\pi, \\
&& 0< -y^5+y^6 <\pi, \\
&& 0< \sqrt 2 y^7 <\pi,
\end{eqnarray}
and characterize completely the range for $\vec y$.

%%%%%%%%%%%%%%%%%%%%%%%%%%%%%%%%%%%%%%%%%%%%%%%%%%%%%%%%%%%%%%%%%%%%%%%%%%%%%%%%%%%%%%%%%%%%%%%%%%%%%%%%%%%%%%%%%%%%%%%%%%%%%
%%%%%%%%%%%%%%%%%%%%%%%%%%%%%%%%%%%%%%%%%%%%%%%%%%%%%%%%%%%%%%%%%%%%%%%%%%%%%%%%%%%%%%%%%%%%%%%%%%%%%%%%%%%%%%%%%%%%%%%%%%%%%
\section{The $E_{7(-5)}$ construction}\label{sec:E7(-5)}
The $E_{7(-5)}$ construction can be easily obtained from the $E_{7(-25)}$ one by following the analysis of Yokota in \cite{IY}.
The maximal compact subgroup of $E_{7(-5)}$ is $U_5=({\rm Spin (12)}\times {\rm SU}(2))/(\mathbb{Z}_2 \times \mathbb{Z}_2)$.
Let us consider the map
\begin{eqnarray}
  \sigma: \jj \longrightarrow \jj, \quad\ \begin{pmatrix}
    a & o_1 & o_2 \\
    \bar o_1 & b & o_3 \\
    \bar o_2 & \bar o_3 & c
  \end{pmatrix} \longmapsto
  \begin{pmatrix}
    a & -o_1 & -o_2 \\
    -\bar o_1 & b & o_3 \\
    -\bar o_2 & \bar o_3 & c
  \end{pmatrix}.
\end{eqnarray}
Thus $\sigma$ lies in the group $F_4 \subset E_6 \subset E_7$. In \cite{IY} it is shown that $U_5$ is the subgroup of elements
$g\in E_7$ such that $\sigma g=g \sigma$. \\
Let us go back to the Lie algebra. In section \ref{sec:tits1} we have realized the Lie algebra $E_{7}$ as $E_6+i\rr+ \jj_{\mathbb C}$, where
$\rr$ is generated by the derivation $D_{i}$ associated to the imaginary unit $i$, and
$\jj_{\mathbb C}\simeq \jj \oplus i\jj \simeq (\rr D_j\oplus j \otimes \jj') \oplus (\rr D_k \oplus k\otimes \jj')$ is the complexification
of the exceptional Jordan algebra $\jj$. Then, it follows that the Lie algebra of $U_5$ is generated by the subgroup
$E_6^\sigma =\{ h\in E_6| \sigma h=h \sigma \} \simeq {\rm spin}(10)\times i\mathbb R$ plus the elements $\jj_{\cc}^\sigma$ which are invariant under $\sigma$.
$E_6^\sigma$ is the subset of elements of the Lie algebra $E_6$ that leave the element
$$
J_1:=\begin{pmatrix}
    1 & 0 & 0 \\
    0 & 0 & 0 \\
    0 & 0 & 0
  \end{pmatrix}
$$
invariant. It is composed by the $A\in E_6$ such that $A J_1=0$, which generate a ${\rm spin} (10)$ subalgebra, plus the $B\in E_6$ with $BJ_1=ib J_1$ for some
real $b$, generating the $U(1)$ factor $i\rr$. With respect to the basis we have chosen for $\jj$, the element $J_1$ is
\begin{eqnarray}
J_1=\frac 16 (3j_1+\sqrt 3 j_{18} +\sqrt 6 j_{27}).
\end{eqnarray}
Thus, the matrices of $E_6$ generating ${\rm spin} (10)$ are the ones having
\begin{eqnarray}
\vec v =\frac 16 (3\vec e_1 +\sqrt 3 \vec e_{18} + \sqrt 6 \vec e_{27})
\end{eqnarray}
in the kernel ($e_i$ are the canonical vectors of $\rr^{27}$). These are\footnote{We reorder the numbering of the matrices by calling them $L_b$
following the natural order. So, for example, $L_1=Y_4, L_2=Y_5, L_{22}=Y_{33}$ and so on. Moreover, when needed we will add multiplicative factors chosen
so that the $L_a$ have the same normalization as the $Y_i$.}
\begin{eqnarray}
  L_b=Y_a, L_{45}=\frac 12 (\sqrt 3 Y_{73}-Y_{56}),  \qquad\ a=4,\ldots, 24, 33,\ldots, 39, 48,\ldots, 55, 74,\ldots, 81, \quad\ b=1,\ldots, 44.
\end{eqnarray}
For the $U(1)$ factor in $E_6^\sigma$ from the above condition one finds the generator $(\sqrt 3 Y_{56}+Y_{73})/2$. However, we can obtain a $U(1)$ factor
also by adding any multiple of $Y_1$. In order to get an orthogonal basis (in particular orthogonal to the next generators) it is convenient to take
the generator
\begin{eqnarray}
  L_{46}=\frac 16 (3Y_{56}+\sqrt 3 Y_{73} +2\sqrt 6 Y_{1}).
\end{eqnarray}
In order to construct the ${\rm SU}(2)$ factor, note that $J_1$ is embedded
in $\jj_\cc$ and generates the real spaces $\rr j\otimes J_1$ and $\rr k\otimes J_1$, which algebraically generate the ${\rm SU}(2)$ factor in $U_5$ \cite{IY}.
As $J_1$ correspond to the vector $\vec v$ above in our basis it follows that the corresponding generators of ${\rm su}(2)$ are
\begin{eqnarray}
L_{67}=\frac 1{3\sqrt 2} (3Y_{82}+\sqrt 3 Y_{99} +\sqrt 6 Y_{2}), \qquad L_{68}=\frac 1{3\sqrt 2} (3Y_{108}+\sqrt 3 Y_{125} +\sqrt 6 Y_{3}),
\qquad L_{69}=\frac 1{3\sqrt 2} (3Y_{56}+\sqrt 3 Y_{73} -\sqrt 6 Y_{1}).
\end{eqnarray}
Finally, we are left with the 20 generators in ${\jj}_\cc^\sigma$ which are complementary to $J_1$. These are
\begin{eqnarray}
&& L_{47}=\frac 1{3\sqrt 2} (-3Y_{82}+\sqrt 3 Y_{99} +\sqrt 6 Y_{2}), \qquad L_{48}=\frac 1{3\sqrt 2} (-3Y_{108}+\sqrt 3 Y_{125} +\sqrt 6 Y_{3}), \cr
&& L_{49}=\frac {\sqrt 2}3 (-\sqrt 3 Y_{99} +\sqrt {\frac 32} Y_{2}), \qquad L_{50}=\frac {\sqrt 2}3 (-\sqrt 3 Y_{125} +\sqrt {\frac 32} Y_{3}), \cr
&& L_d=Y_b, \qquad b=100,\ldots, 107, 126,\ldots, 133, \quad\ d=51,\ldots, 66.
\end{eqnarray}
In this way we have selected the ${\rm Spin}(12)\times {\rm SU}(2)$ subalgebra and we can obtain the $E_{7(-5)}$ real form by applying the unitary Weyl trick
to the complementary generators
\begin{eqnarray}
&& L_e=Y_c, \qquad c=25,\ldots, 32, 40,\ldots, 47, 57,\ldots, 72, 83,\ldots, 98, 109,\ldots, 124, \quad\ e=70,\ldots,133.
\end{eqnarray}
%%%%%%%%%%%%%%%%%%%%%%%%%%%%%%%%%%%%%%%%%%%%%%%%%%%%%%%%%%%%%%%%%%%%%%%%%%%%%%%%%%%%
\subsection{Construction of the group.}
The symmetric manifold $E_{7(-5)}/U_5$ has rank $4$, which means that we can find a Cartan subalgebra of $E_7$ with four generators in the complement
of $\lie(U_5)$. A possible choice for such a complement is
\begin{eqnarray}
H_4:=\langle L_{70}, L_{86}, L_{103}, L_{120} \rangle.
\end{eqnarray}
Let us indicate with $\mathfrak k$ the maximal Lie subalgebra of $\mathfrak u_5:= \lie (U_5)$ that commute with $H_4$. It results that the corresponding
Lie group is $K \simeq {\rm Spin}(4)\times {\rm SU}(2) \subset {\rm Spin}(12) \times \mathbb{Z}_2^4 \subset U_5$, se appendix \ref{app:ultimissima/e/basta}.
This means that we can write the group elements in the form
\begin{eqnarray}
  E_7[x_1,\ldots, x_{60};y_1,\ldots,y_4;z_1,\ldots,z_{69}]=(U_5/K)[x_1,\ldots, x_{60}] e^{H_4[y_1,\ldots,y_4]} U_5[z_1,\ldots,z_{69}],
\end{eqnarray}
where
\begin{eqnarray}
H_4[y_1,\ldots,y_4]:= y_1 L_{70}+y_2 L_{86} +y_3 L_{103} +y_4 L_{104},
\end{eqnarray}
and $U_5[z_1,\ldots,z_{69}]$ and $(U_5/K)[x_1,\ldots, x_{60}]$ are parameterizations of $U_5$ and $U_5/K$ respectively. Again we will now determine the range
for the parameters $y_i$, the ranges for $x_a$ and $z_b$ being determinable as usual.\\
Let $\mathfrak{p}$ the orthogonal complement of $\mathfrak{u}_5$ in $\lie E_7$, $\mathfrak t$ the orthogonal complement of $H_4$ in $\mathfrak p$,
and $\mathfrak s$ the orthogonal complement of $\mathfrak k$ in ${\mathfrak u}_5$.
The invariant measure over $E_7$ in this construction is
\begin{eqnarray}
&& d\mu_{E_7}=d\mu_{U_5/K} [\underline x]\ d\mu_{U_5} [\underline z]\ |h(\vec y)|\ d\vec y^4, \\
&& h(\vec y ):={\rm det} [\Pi \circ {\rm Ad}_{e^{-H_4}} : \mathfrak{s} \rightarrow \mathfrak {t} ],
\end{eqnarray}
where $\Pi$ is the orthogonal projection on $\mathfrak t$. Again, we assume to work with orthonormal bases, so that the determinant function
is well defined. Using a method similar to the one used in appendix \ref{app:range77}, one gets
\begin{eqnarray}
  |h(\vec y)|=\prod_{\beta \in {\rm Rad'}^+} \sin^{m_\beta} (|\beta(H_4(y_1,\ldots,y_4))| ),
\end{eqnarray}
where ${\rm Rad'}^+$ is the set of positive restricted roots of $E_7/U_5$ w.r.t. $H_4$, and $m_\beta$ is the multiplicity of $\beta$. The restricted
root lattice of $E_7/U_5$ is an $F_4$ lattice whose short roots have multiplicity $4$.
Thus, the equations for the range of $\vec y$ are
\begin{eqnarray}
0< |\beta(H_4(y_1,\ldots,y_4))|<\pi, \qquad\ \beta \in {\rm Rad}^+.
\end{eqnarray}
This is a set of $24$ double inequalities, which, however, can be quickly reduced to a set of five equations as follows. Indeed, all positive
roots can be obtained as non negative integer linear combinations of the simple roots. The simple roots are
\begin{eqnarray}
\alpha_1=\frac 12 (1,-1,-1,1), \qquad \alpha_2=(0,0,1,-1), \qquad \alpha_3=(-1,0,0,0), \qquad \alpha_4=(0,1,-1,0).
\end{eqnarray}
In particular, there exists a unique longest root
$$\beta_{max}=4\alpha_1+3\alpha_2+2\alpha_3+2\alpha_4$$
whose coefficients are the highest ones. From this it follows that all inequalities are a consequence of the ones corresponding to the simple roots plus the
one associated to the longest one. These are
\begin{eqnarray}
&& 0< \frac 12 (y^1-y^2-y^3+y^4) <\pi, \\
&& 0< y^3-y^4 <\pi, \\
&& 0< y^1 <\pi, \\
&& 0< y^2-y^3 <\pi, \\
&& 0< y^3+y^4 <\pi,
\end{eqnarray}
and characterize completely the range for $\vec y$.

%%%%%%%%%%%%%%%%%%%%%%%%%%%%%%%%%%%%%%%%%%%%%%%%%%%%%%%%%%%%%%%%%%%%%%%%%%%%%%%%%%%%%%%%%%%%%%%%%%%%%%%%%%%%%%%%%%%%%%%%%%%%%
\section{Conclusions}\label{sec:conc}
In this paper we have solved the problem of giving an explicit construction of the exceptional $E_7$ simple Lie group and in particular of its compact form.
We have solved the problem in three different ways. In the first one we have realized the generalized Euler construction w.r.t. the
maximal compact subgroup $U=(E_6\times U(1))/\mathbb{Z}_3$ of highest dimension.
To this end we have first obtained the adjoint representation of the $E_7$ Lie algebra
by using the Tits realization of the Magic Squares. This has allowed us to easily understand
the structure of the commutators and then the main properties of the algebra in relation to the subalgebra Lie($U$).
However, since the Adjoint representation of the group has a nontrivial kernel given by the center $\mathbb {Z}_2$ of $E_7$, we have  built  also
the fundamental representation ${\bf 56}$ which provides a faithful representation of
the group. For this case we have reported a very careful analysis.
Note that we can obtain a realization of $E_7/\mathbb{Z}_2$, simply by replacing the matrices in the representation ${\bf 56}$ with those in
the ${\bf 133}$. In this case, the center is mapped into $\mathbb{I}_{133}$, so that we need to restrict the range of the
$U(1)$ parameter $x_{55}$  to the interval $[0, \sqrt {\frac 23}\ \pi]$.\\
Since there exists a generalized Euler parametrization for each maximal compact subgroup, we have worked out all such parametrization. Indeed, the other two
correspond to the maximal compact subgroups  $\tilde U:={\rm SU}(8)/\mathbb{Z}_2$ and $U_5=({\rm Spin}(12)\times {\rm SU}(2))/(\mathbb{Z}_2\times \mathbb{Z}_2)$.
This is not merely an exercise. Indeed, it can be relevant to be able to recognize a specific subgroup for a given application. Moreover
from each Euler construction one can obtain the corresponding real form simply by means of the unitary Weyl trick. For these constructions
we have omitted several details, in part to avoid annoying repetitions and in part because certain specific steps revealed to have a deeper meaning
which can be understood in a more general context, which will be presented in a devoted paper \cite{to-appear}.

\vspace{5cm}

\subsection*{Acknowledgments.} We are grateful to B.L. Cerchiai and A. Marrani for very useful comments.

\newpage

%%%%%%%%%%%%%%%%%%%%%%%%%%%%%%%%%%%%%%%%%%%%%%%%%%%%%%%%%%%%%%%%%%%%%%%%%%%%%%%%%%%%%%%%%%%%%%%%%%%%%%%%%%%%%%%%%%%%%%%%%%%%%%%%
\begin{appendix}
\section{The Tits product}\label{app:a}
We will follow (\cite{bart-sud}), and, for generality, we will allow for $\hh$ to be the non associative octonionic algebra.\\
Let $k\otimes j, k'\otimes j', k''\otimes j'' \in \hh' \otimes \jj'$ and define
\eqn
[h\otimes j, h'\otimes j']:= \frac \alpha{3} \langle j,j' \rangle D_{h,h'}-\beta \langle h,h'\rangle [L_j, L_{j'}]
+\gamma [h,h']\otimes (j\star j')
\feqn
for some constants $\alpha, \beta, \gamma$. Then, the only nontrivial Jacobi identities to be checked are
\eqn
&& 0=[[h\otimes j, h'\otimes j'],h''\otimes j'']+[[h''\otimes j'', h\otimes j],h'\otimes j']+[[h'\otimes j', h''\otimes j''],h\otimes j]\cr
&&\quad =:{\rm cyc}\{ [[h\otimes j, h'\otimes j'],h''\otimes j''] \}.
\feqn
Now
\eqn
[[h\otimes j, h'\otimes j'],h''\otimes j''] && =\frac {\alpha}3 \langle j, j' \rangle D_{k,k'} (k'') \otimes j''
-\beta \langle k,k' \rangle k''\otimes [L_{j}, L_{j'}] j'' +\frac {\alpha\gamma}3 \langle j\star j', j''\rangle D_{[k,k'],k''}\cr
&& -\gamma\beta \langle [k,k'], k''\rangle \otimes [L_{j\star j'},L_{j''}]+\gamma^2 [[k,k'],k''] \otimes (j\star j')\star j''.
\feqn
After cyclic summation the terms with coefficients $\alpha\gamma$ and $\beta\gamma$ disappear.
Defining the associator $[k,k',k'']=(kk')k''-k(k'k'')$ (which is totally antisymmetric)
we have the identities, which can be easily obtained adapting the
relations in \cite{bart-sud} to our notation,
\eqn
&& D_{k,k'} (k'')=[[k,k'],k'']-3[k,k',k''], \\
&& \langle c,a \rangle b-\langle c,b \rangle a =-\frac 14 [[b,a],c]+\frac 12 [c,b,a],\\
&& X\star (X\circ X)=\frac 12 \langle X,X \rangle X, \qquad X\in \jj'.
\feqn
From the last one we also get
\eqn
j\star (j'\circ j'')-\frac 12  \langle j,j' \rangle j'' +{\rm cyclic}=0.
\feqn
Using all this identities we finally get
\begin{align}
{\ {\rm cyc}}\{ [[h\otimes j, h'\otimes j'],h''\otimes j'']\}=&\left\{ [[k,k'],k''] \otimes \left[ \frac {\alpha-\gamma^2}3 \langle j,j'\rangle j''
 +\frac {4\gamma^2-\beta}4 (j\circ j')\star j'' \right] \right\} \cr
&+\frac {\beta-4\alpha}4 [k,k',k'']\otimes {\ {\rm cyc}} [\langle j,j' \rangle j''],
\end{align}
which vanishes for all choices of $k,k',k''$ and $j,j',j''$ if and only if $\alpha=\gamma^2=\frac \beta4$, which gives equation (\ref{prod}).

%%%%%%%%%%%%%%%%%%%%%%%%%%%%%%%%%%%%%%%%%%%%%%%%%%%%%%%%%%%%%%%%%%%%%%%%%%%%%%%%%%%%%%%%%%%%%%%%%%%%%%%%%%%%%%%%%%%%%%%%
%%%%%%%%%%%%%%%%%%%%%%%%%%%%%%%%%%%%%%%%%%%%%%%%%%%%%%%%%%%%%%%%%%%%%%%%%%%%%%%%%%%%%%%%%%%%%%%%%%%%%%
\section{The adjoint representation of $E_7$ in the $E_{7(-25)}$ construction}\label{sec:adjoint}
We can now realize the adjoint representation of $E_7$.
To this end, let us consider the basis $\{\Psi_A \}_{A=1}^{133}$ of $\mathfrak{g}$, with
\eqn
\Psi_A=\left\{
\begin{array}{lll}
H_L & \mbox{if} & A=1,2,3, \quad L=1,2,3 \\
J_{I} & \mbox{if} & A=I+3, \quad\ I=1,\ldots,52, \\
h_1\otimes j_{a} & \mbox{if} & A=a+55, \quad\ a=1,\ldots, 26, \\
h_2\otimes j_{a} & \mbox{if} & A=a+81, \quad\ a=1,\ldots, 26, \\
h_3\otimes j_{a} & \mbox{if} & A=a+107, \quad\ a=1,\ldots, 26.
\end{array}
\right.
\feqn
The dual basis is $\hat{\Psi}^B$, with $\hat{\Psi}^B(\Psi_A)=\delta^B_A$.
Thus, the generic matrix of $\lie (E_7)$ is $(M_A)^C_{\phantom{C}B}=\hat{\Psi}^C(ad_{\Psi_A}\Psi_B)$ where $ad_{\Psi_A}\Psi_B=[\Psi_A,\Psi_B]$
is the adjoint action of
$\Psi_A$ on $\Psi_B$. As usual, in $(M_A)^C_{\phantom{C}B}$, the upper index denote the row, the lower the column and the index $A$ labels
the matrix we are considering. We will use the first capital Latin letters $A,B,C$, running from 1 to 133, to label the basis elements of $\lie(E_7)$,
small Latin letters $i,j,k$, running from 1 to 3, to label the imaginary units of $\hh$, capital Latin letters $L,M,N$, running from 1 to 3, for
the derivations $D(\hh)$, capital Latin letters $I,J,K$, running from 1 to 52, for the derivation $D(\jj)$, first small
Latin letters $a,b,c$, running from 1 to 26, for the elements of $\jj '$ and finally the Greek letters $\mu, \nu, \lambda$, running from 1 to 27,
for the elements of $\jj$.\\
Let us proceed with the construction step by step.
%%%%%%%%%%%%%%%%%%%%%%%%%%%%%%%%%%%%%%%%%%%%%%%%%%%%%%%%%%%%%%%%%%%%%%%%%
\subsection{The matrices $M_A$ with $A=1,2,3$}
These matrices are determined by the adjoint action of the elements $\Psi_A= H_L$, with $L=1,2,3$ on  $\Psi_B$, $B=1,\ldots,133$.
We have to consider five cases for the index $B$.
\begin{itemize}
\item $B=1,2,3$, i.e. $M=1,2,3$: $[H_L,H_M]=2\epsilon_{LM}^{\phantom{LM}N}H_N$.
\item $B=4,\cdots,55$, i.e. $J=1,\cdots,52$: $[H_L,J_J]=0$.
\item $B=56,\cdots,81$, i.e. $b=1,\cdots,26$: $[H_L,h_1\otimes j_b]=ad_{h_i}h_1\otimes j_b=2\epsilon_{i1}^{\phantom{i1}k}h_k\otimes j_b$,
with the value of $L$ equal to that of $i$.
\item $B=82,\cdots,107$, i.e. $b=1,\cdots,26$: $[H_L,h_2\otimes j_b]=ad_{h_i}h_2\otimes j_b=2\epsilon_{i2}^{\phantom{i2}k}h_k\otimes j_b$,
with the value of $L$ equal to that of $i$.
\item $B=108,\cdots,133$, i.e. $b=1,\cdots,26$: $[H_L,h_3\otimes j_b]=ad_{h_i}h_3\otimes j_b=2\epsilon_{i3}^{\phantom{i3}k}h_k\otimes j_b$,
with the value of $L$ equal to that of $i$.
\end{itemize}
Applying the dual basis $\hat{\Psi}^C$,
and organizing the result in block matrices (with block structure $C=\{N,K,a,b,c\}$) we obtain the first three matrices:
\eqn
&& (M_1)^C_{\phantom{C}B}=
\left(
\begin{array}{c|c|c|c|c}
\begin{array}{ccc}
0&0&0 \\
0&0&-2 \\
0&2&0 \\
\end{array}
& 0 & 0 & 0 & 0 \\
\hline
0&0& 0 & 0 & 0 \\
\hline
0& 0 & 0 & 0 & 0 \\
\hline
0& 0 & 0 & 0 & - 2{I_{26}} \\
\hline
0 & 0 & 0 & 2{I_{26}} & 0 \\
\end{array}
\right), \qquad
(M_2)^C_{\phantom{C}B}=
\left(
\begin{array}{c|c|c|c|c}
\begin{array}{ccc}
0&0&2 \\
0&0&0 \\
-2&0&0 \\
\end{array}
& 0 & 0 & 0 & 0 \\
\hline
0&0& 0 & 0 & 0 \\
\hline
0& 0 & 0 & 0 & 2{I_{26}} \\
\hline
0& 0 & 0 & 0 & 0 \\
\hline
0 & 0 & -2{I_{26}} & 0 & 0\\
\end{array}
\right)\cr
&& (M_3)^C_{\phantom{C}B}=
\left(
\begin{array}{c|c|c|c|c}
\begin{array}{ccc}
0&-2&0 \\
2&0&0 \\
0&0&0 \\
\end{array}
& 0 & 0 & 0 & 0 \\
\hline
0& 0 & 0 & 0 & 0\\
\hline
0& 0 & 0 & -2{I_{26}} & 0 \\
\hline
0 & 0 & 2{I_{26}} & 0 & 0\\
\hline
0&0& 0 & 0 & 0 \\
\end{array}
\right),
\feqn
where the five diagonal blocks have dimensions $3\times 3$, $52\times 52$, $26\times 26$, $26\times 26$ and $26\times 26$ respectively.

%%%%%%%%%%%%%%%%%%%%%%%%%%%%%%%%%%%%%%%%%%%%%%%%%%%%%%%%%%%%%%%%%%%%%%%%%%%%%%%%%%%%%%%%%%%%%%%%%%%

\subsection{The matrices $M_A$ with $A=4,\cdots,55$}
These matrices are given by the adjoint action of the elements $\Psi_{I+3}= J_I$, with $I=1,\cdots, 52$ on $\Psi_B$, $B=1,\ldots,133$.
\begin{itemize}
\item $B=1,2,3$, i.e. $M=1,2,3$: $[J_I,H_M]=0$.
\item $B=4,\cdots,55$, i.e. $J=1,\cdots,52$: $[J_I,J_J]=f_{IJ}^{\phantom{IJ}K}J_K$, where $f_{IJ}^{\phantom{IJ}K}$ are the structure constants of $\lie(F_4)$.
\item $B=56,\cdots,81$, i.e. $b=1,\cdots,26$: $[J_I,h_1\otimes j_b]=h_1\otimes J_Ij_b=h_1\otimes(C_I)^c_{\phantom{c} b}j_c$. The $(C_I)^c_{\phantom{I} b}$ are the matrices of $\lie(F_4)$, the algebra of the derivation of the octonionic Jordan matrices.
    Note that $j_b$ are traceless, actually we choose as basis for $\jj$ 26 traceless matrices and the identity. As derivations vanish on the identity, $\lie(F_4)$ can be obtained as the derivations on $\jj'$.
\item $B=82,\cdots,107$, i.e. $b=1,\cdots,26$: $[J_I,h_2\otimes j_b]=h_2\otimes J_Ij_b=h_2\otimes(C_I)^c_{\phantom{c} b}j_c$.
\item $B=108,\cdots,133$, i.e. $b=1,\cdots,26$: $[J_I,h_3\otimes j_b]=h_3\otimes J_Ij_b=h_3\otimes(C_I)^c_{\phantom{c} b}j_c$.
\end{itemize}
We obtain:
\be
(M_{I+3})^C_{\phantom{C}B}=
\left(
\begin{array}{c|c|c|c|c}
0& 0 & 0 & 0 & 0 \\
\hline
0& f_{IJ}^{\phantom{IJ}K}& 0 & 0 & 0 \\
\hline
0& 0 & (C_I)^c_{\phantom{c}b} & 0 & 0 \\
\hline
0& 0 & 0 & (C_I)^c_{\phantom{c}b} & 0 \\
\hline
0 & 0 & 0 & 0 & (C_I)^c_{\phantom{c}b} \\
\end{array}
\right),
\ee
$I=1,\ldots,52$. % where we recall that upper indexes label rows.
%%%%%%%%%%%%%%%%%%%%%%%%%%%%%%%%%%%%%%%%%
\subsection{The matrices $M_A$ with $A=56,\cdots,81$}
These matrices are given by the adjoint action of the elements $\Psi_A= h_1\otimes j_a$, with $a=1,\cdots, 26$ on $\Psi_B$, $B=1,\ldots,133$.
\begin{itemize}
\item $B=1,2,3$, i.e. $M=1,2,3$: $[h_1\otimes j_a ,H_M]=-[H_M,h_1\otimes j_a]=-ad_{h_j}h_1\otimes j_a=-2\epsilon_{j1}^{\phantom{j1}k}h_k\otimes
      j_a$, with the value of $M$ equal to that of $j$.
\item $B=4,\cdots,55$, i.e. $J=1,\cdots,52$: $[h_1\otimes j_a,J_J]=-[J_J,h_1\otimes j_a]=-h_1\otimes J_Jj_a=-h_1\otimes
    (C_J)^c_{\phantom{b}a}j_c$.
\item $B=56,\cdots,81$, i.e. $b=1,\cdots,26$: $[h_1\otimes j_a,h_1\otimes j_b]=\frac 1{12}
    H_{[h_1,h_1]}<j_a,j_b>-<h_1,h_1>[L_{j_a},L_{j_b}]+\frac 12[h_1,h_1]\otimes (j_a\star
    j_b)=-<h_1,h_1>[L_{j_a},L_{j_b}]=-[L_{j_a},L_{j_b}]=-\alpha_{ab}^{\phantom{ab}K}J_K$.
The last equality follows from
$[L_{j_a},L_{j_b}]j_\mu=L_{ja}(L_{j_b}(j_\mu))-L_{j_b}(L_{j_a}(j_\mu))$.
The Jordan algebra is commutative so that the left action and the right action on $j_\mu$ coincide. On the other hand, the right action of $\jj'$ on
$\jj$ gives the 26 new elements of the extension the Lie algebra of $F_4$ to $\lie(E_6)$, so that $[L_{j_a},L_{j_b}]$ can be expressed in terms of the
26 $\lie(E_6)$
matrices that do not belong to $\lie(F_4)$:
$[L_{j_a},L_{j_b}]j_\mu=j_\lambda[\tilde{C}_a^{\phantom{a}\lambda}{}_\nu\tilde{C}_b^{\phantom{b}\nu}{}_\mu-\tilde{C}_b^{\phantom{b}
\lambda}{}_\nu\tilde{C}_a^{\phantom{a}\nu}{}_\mu]=j_\lambda[\tilde{C_a},\tilde{C_b}]^\lambda_{\phantom{\lambda}\mu}=
j_\lambda\alpha_{ab}^{\phantom{ab}K}{C_K}^\lambda_{\phantom{\lambda}\mu}=\alpha_{ab}^{\phantom{ab}K}J_Kj_\mu=
\alpha_{ab}^{\phantom{ab}K}J_Kj_c+\alpha_{ab}^{\phantom{ab}K}J_Kj_{27}=\alpha_{ab}^{\phantom{ab}K}J_Kj_c$,
because a derivation on the identity vanishes.
The third equality follows from the fact that the commutator between two of the 26 matrices of $\lie(E_6)/\lie(F_4)$,
lies in $\lie(F_4)$.
\item $B=82,\cdots,107$, i.e. $b=1,\cdots,26$: $[h_1\otimes j_a,h_2\otimes j_b]= [h_1\otimes j_a,h_1\otimes j_b]=\frac 1{12}
 H_{[h_1,h_2]}<j_a,j_b>-<h_1,h_2>[L_{j_a},L_{j_b}]+\frac 12[h_1,h_2]\otimes (j_a\star j_b)=\frac 1{12}
 H_{[h_1,h_2]}\tau\delta_{ab}+\frac 12 [h_1,h_2]\otimes
 (j_a\star j_b)=\frac 16\tau\delta_{ab}H_3+h_3\otimes j_c(\tilde{C}_a)^c_{\phantom{c}b}$. The last equality follows from
 $(j_a\star j_b)=j_a\circ j_b-\frac 13<j_a,j_b>I_3=R_{j_a}(j_b)-\frac 13<j_a,j_b>I_=(\tilde{C}_a)^\mu_{\phantom{\mu}b}j_\mu-
 \frac 13\tau\delta_{ab}I_3=(\tilde{C}_a)^c_{\phantom{c}b}j_c+(\tilde{C}_a)^{27}_{\phantom{27}b}j_{27}-
 \frac 13\tau\delta_{ab}I_3=(\tilde{C}_a)^c_{\phantom{c}b}j_c$.
The last two terms cancel because the $\star$ product gives traceless elements by definition. This fact can also be verified using the explicit
expression of the matrices $\tilde{C}_a$.
\item $B=108,\cdots,133$, i.e. $b=1,\cdots,26$: $[h_1\otimes j_a,h_3\otimes j_b]=-\frac 16\tau\delta_{ab}H_2-h_2\otimes
    j_c(\tilde{C}_a)^c_{\phantom{c}b}$.
\end{itemize}
Collecting all these results, we get
\be
(M_{a+55})^C_{\phantom{C}B}=
\left(
\begin{array}{c|c|c|c|c}
0& 0 & 0 & \frac 16\tau\delta_{3}^N\delta_{ab} & -\frac 16\tau\delta_{2}^N\delta_{ab} \\
\hline
0& 0& -\alpha_{ab}^{\phantom{ab}K} & 0 & 0 \\
\hline
0&  -(C_J)^c_{\phantom{c}a} & 0 & 0 & 0\\
\hline
-2\delta_{a}^c\delta_{3M} & 0 & 0 & 0 & -(\tilde{C}_a)^c_{\phantom{c}b} \\
\hline
2\delta_{a}^c\delta_{2M}  & 0 & 0 & (\tilde{C}_a)^c_{\phantom{c}b} & 0 \\
\end{array}
\right),
\ee
$a=1,\ldots, 26$.

%%%%%%%%%%%%%%%%%%%%%%%%%%%%%%%%%%%%%%%%%%%%%%%%%%%%%%%%%%%%%%%%%%%%%%%%%%%%%%%%%%%%%
\subsection{The matrices $M_A$ with $A=82,\cdots,133$}
With analogous computations as the previous case we get the last 52 matrices:
\eqn
&&(M_{a+81})^C_{\phantom{C}B}=
\left(
\begin{array}{c|c|c|c|c}
0& 0 & -\frac 16\tau\delta_{3}^N\delta_{ab} & 0 & \frac 16\tau\delta_{1}^N\delta_{ab} \\
\hline
0& 0& 0& -\alpha_{ab}^{\phantom{ab}K} & 0\\
\hline
2\delta_{a}^c\delta_{3M} & 0 & 0 & 0 & (\tilde{C}_a)^c_{\phantom{c}b}\\
\hline
0  & -(C_J)^c_{\phantom{c}a} & 0 & 0 & 0\\
\hline
-2\delta_{a}^c\delta_{1M}  & 0 & -(\tilde{C}_a)^c_{\phantom{c}b} & 0 & 0\\
\end{array}
\right),
\feqn
\eqn
&&(M_{a+107})^C_{\phantom{C}B}=
\left(
\begin{array}{c|c|c|c|c}
0& 0 & \frac 16\tau\delta_{2}^N\delta_{ab} & -\frac 16\tau\delta_{1}^N\delta_{ab} & 0\\
\hline
0& 0& 0& 0& -\alpha_{ab}^{\phantom{ab}K} \\
\hline
-2\delta_{a}^c\delta_{2M} & 0 & 0 & -(\tilde{C}_a)^c_{\phantom{c}b} & 0\\
\hline
2\delta_{a}^c\delta_{1M}  & 0 & (\tilde{C}_a)^c_{\phantom{c}b} & 0 & 0\\
\hline
0  & -(C_J)^c_{\phantom{c}a} & 0 & 0 & 0\\
\end{array}
\right),
\feqn
$a=1,\ldots, 26$.
This completes the construction of the representation $\bf{133}$.
%As we will see, to construct the group it is convenient
%to realize also  the $\bf {56}$ .

%%%%%%%%%%%%%%%%%%%%%%%%%%%%%%%%%%%%%%%%%%%%%%%%%%%%%%%%%%%%%%%%%%%%%%%%%%%%%%%
\section{The representation \boldmath{$56$} of $E_7$}\label{sec:56}
%To this end, we will follow \cite{IY}, section 7.\\
By applying the Yokota method we obtain the following 133 $56\times 56$ matrices:
\eqn
&& Y_1=\left(
\begin{array} {c|c|c|c}
\frac i{\sqrt 6} I_{27}  & \ov &  \oo & \ov \cr
\hline
^t\ov  & -i{\sqrt {\frac 32}} & \ov^t & 0 \cr \hline
\oo  & \ov & -\frac i{\sqrt 6} I_{27} & \ov \cr \hline
^t\ov  & 0 & \ov^t & i{\sqrt {\frac 32}}
\end{array}
\right),
\feqn
\eqn
&& Y_2=\frac 12 \left(
\begin{array} {c|c|c|c}
\oo  & \ov & -i\sqrt {\frac 23} \tilde {I} & i\sqrt 2 \vec e_{27} \cr
\hline
\ov^t  & 0 & i\sqrt 2\ \vec e_{27}^t & 0 \cr
\hline
-i\sqrt {\frac 23} \tilde {I}  & i\sqrt 2 \vec e_{27} & \oo & \ov \cr
\hline
i\sqrt 2\ \vec e_{27}^t  & 0 & \ov^t & 0
\end{array}
\right),
\feqn
\eqn
&& Y_3=\frac 12 \left(
\begin{array} {c|c|c|c}
\oo  & \ov & \sqrt {\frac 23} \tilde {I} & \sqrt 2 \vec e_{27} \cr
\hline
\ov^t  & 0 & \sqrt 2\ \vec e_{27}^t & 0 \cr
\hline
-\sqrt {\frac 23} \tilde {I}  & -\sqrt 2 \vec e_{27} & \oo & \ov \cr
\hline
-\sqrt 2\ \vec e_{27}^t  & 0 & \ov^t & 0
\end{array}
\right),
\feqn
\eqn
&& Y_{I+3}=\left(
\begin{array} {c|c|c|c}
\phi_I  & \ov & \oo & \ov \cr
\hline
\ov^t  & 0 & \ov^t & 0 \cr
\hline
\oo  & \ov & -\phi_I^t & \ov \cr
\hline
\ov^t  & 0 & \ov^t & 0
\end{array}
\right), \qquad\ I=1, \ldots, 78,
\feqn
\eqn
&& Y_{a+81}=\frac 12 \left(
\begin{array} {c|c|c|c}
\oo  & \ov & 2iA_a & i\sqrt 2 \vec e_a \cr
\hline
\ov^t  & 0 & i\sqrt 2 \vec e_a^t & 0 \cr
\hline
2iA_a  & i\sqrt 2 \vec e_a & \oo & \ov \cr
\hline
i\sqrt 2 \vec e_a^t  & 0 & \ov^t & 0
\end{array}
\right), \qquad\ a=1, \ldots, 26,
\feqn
\eqn
&& Y_{a+107}=\frac 12 \left(
\begin{array} {c|c|c|c}
\oo  & \ov & -2A_a & \sqrt 2 \vec e_a \cr
\hline
\ov^t  & 0 & \sqrt 2 \vec e_a^t & 0 \cr
\hline
2A_a  & -\sqrt 2 \vec e_a & \oo & \ov \cr
\hline
-\sqrt 2 \vec e_a^t  & 0 & \ov^t & 0
\end{array}
\right), \qquad\ a=1, \ldots, 26,
\end{eqnarray}
where $I_{n}$ is the $n\times n$ identity matrix, $\oo$ is the $27\times 27$ null matrix, $\vec 0_{n}$ is the zero vector in
$\mathbb{R}^{n}$, $\vec e_\mu$, $\mu=1,\ldots, 27$, is the canonical basis of $\mathbb{R}^{27}$,
\begin{eqnarray}
\tilde {I}=\left(
\begin{array}{c|c}
I_{26} & \vec 0_{26} \cr \hline
^t\vec 0_{26} & -2
\end{array}
\right),
\end{eqnarray}
\begin{eqnarray}
\phi_I=\left\{
\begin{array}{c}
C_I \qquad\ I=1,\ldots,52 \cr
\tilde C_{I-52} \qquad\ I=53,\ldots, 78,
\end{array}
\right.
\end{eqnarray}
where $C_I$ and $\tilde C_a$ are defined in (\ref{effe4}) and (\ref{e6}) respectively.
Finally, $A_a$, $a=1,\ldots, 26$ are the $27\times 27$ symmetric matrices representing, by means of the linear isomorphism
$\jj\simeq \mathbb{R}^{27}, \ \ j_\mu \mapsto \vec e_\mu$ , the action of $\jj'$ on $\jj$
defined by
$$
j' (j):= j' \circ j-\frac 12 {\rm Tr}(j) j' -\frac 12 \langle j,j' \rangle I_3.
$$
In particular $j_\mu$ is the basis of $\jj$ defined in the previous section, and $A_a$ is the matrix associated to $j_a$, $a=1,\ldots,26$.\\
The matrices $Y_A$, $A=1,\ldots, 133$ are orthonormalized w.r.t. to the product
$$
\langle Y , Y' \rangle_{\bf{56}} = -\frac 1{12} {\rm Tr} (YY').
$$
It is easy to check, for example by means of a computer, that the $Y_A$ satisfy exactly the same commutation relations of the $M_A$.

%%%%%%%%%%%%%%%%%%%%%%%%%%%%%%%%%%%%%%%%%%%%%%%%%%%%%%%%%%%%%%%%%%%%%%%%%%%%%%%%%%%%%%%%%%%%%%%%%%%%%
\section{The range of the parameters for the $E_{7(-25)}$}\label{sec:parameters}
To determine the range of the parameters, we can proceed as shown in \cite{40}. In particular, we will be able to
make use of the results  already obtained in \cite{F4} and \cite{E6}, thus simplifying most of the computations.\\
Setting $U:=(E_6\times U(1))/\mathbb{Z}_3$, from our previous construction we can write the generic element of $E_7$
as
\begin{eqnarray}\label{euler}
E_7[x_1,\ldots,x_{133}]=B[x_1,\ldots, x_{51}] \exp V[x_{52}, x_{53}, x_{54}] U[x_{55},\ldots,x_{133}],
\end{eqnarray}
which is (\ref{gruppoE7}) realized with the matrices $Y_A$. Then, the range of the parameters $x_{55},\ldots,x_{133}$
can be chosen in such the way that $x_{55}$ covers the whole $U(1)$ and $x_{56}, \ldots, x_{133}$ cover the whole $E_6$.
The last ones have been determined in \cite{E6} and will be reported in the conclusions for convenience. The $U(1)$ is covered
if $x_{55}$ runs over a period $T_g$. However, because of the action of $\mathbb{Z}_3$, we must reduce its range to
$T_g/3=2\sqrt{\frac 23}\pi$.\\
For the remaining parameters, we have to construct the invariant measure over
the quotient $E_7/U$. This is given by
\begin{equation}
d\mu_{B_{E_7}}=|\det J_\p^\perp|,
\end{equation}
where
$$
J_\p=e^{-V} B^{-1} d(B e^V)=dV +e^{-V} B^{-1} dB e^V
$$
and $\perp$ means the part orthogonal to Lie($U$). Concretely, this means that we have to project $J_\p$ on
$Y_2$, $Y_3$, $Y_{82},\ldots,Y_{133}$.
Since
\eqn
dV =dx_{52} Y_2+dx_{53} Y_{82}+dx_{54} Y_{99},
\feqn
we only need to concentrate on the term $e^{-V} B^{-1} dB e^V$, which must be projected on
$Y_3$, $Y_{83},\ldots, Y_{98}$, $Y_{100},\ldots, Y_{133}$.
The details of the computations are given in Appendix \ref{app:measure}. However, to express the result of the computation
we need to look better at the structure of $B=U/U_o$.\\
We can write it as $B=U(1)/\mathbb{Z}_3\times E_6/SO(8)$, where $SO(8)$ is the subgroup $SO(8)\subset F_4 \subset E_6$
which commutes with $Y_{56}:=\psi(c_{53})$ and $Y_{73}:=\psi(c_{70})$, used in \cite{E6} to construct $E_6$. By using
the results in \cite{E6}, \cite{F4}, we then see that
\begin{eqnarray}
B[x_{1},\ldots, x_{51}] =e^{x_1 Y_1}\psi_*(B_{E_6})[x_2,\ldots, x_{27}] \psi_*(B_{F_4})[x_{28},\ldots, x_{43}] \psi_*(B_{SO(9)})[x_{44},\ldots, x_{51}],
\end{eqnarray}
and $B_{E_6}$, $B_{F_4}$ and $B_{SO(9)}$ are the basis in the construction of the groups
\begin{eqnarray}
E_6=B_{E_6} F_4, \qquad\ F_4=B_{F_4} SO(9), \qquad\ SO(9)=B_{SO(9)} SO(8),
\end{eqnarray}
defined in \cite{E6} and \cite{F4} respectively. Also, $\psi$ indicates the map $\psi(c_I)=Y_{I+3}$, $I=1,\ldots,78$, $c_I$ being the generators
of $E_6$. From now on we will omit the map $\psi$, for simplicity.
Starting from this structure, as shown in Appendix \ref{app:measure},
we get
\begin{eqnarray}
&& d\mu_{B_{E_7}}(x_1,\ldots,x_{55})= dx_1 d\mu_{B_{E_6}}(x_2,\ldots, x_{27}) d\mu_{B_{F_4}}(x_{28},\ldots, x_{43})
d\mu_{B_{SO(9)}}(x_{44},\ldots, x_{51})\cr
&& \phantom{d\mu_{B_{E_7}}(x_1,\ldots,x_{55})=}
W(x_{52}, x_{53}, x_{54}) dx_{52} dx_{53} dx_{54}, \label{measure}\\
&& W(x_{1}, x_{2}, x_{3}):= \sin \left( \frac {\sqrt 3 x_2 +\sqrt 2 x_1 +x_3}{\sqrt 3}\right)
\sin \left( \frac {\sqrt 3 x_2 -\sqrt 2 x_1 -x_3}{\sqrt 3}\right)
\sin \left( \frac {\sqrt 6 x_1 -2\sqrt 3 x_3}{3}\right)\cr
&& \phantom{W(x_{1}, x_{2}, x_{3}):=}
\sin \left( \frac {\sqrt 3 x_2 +2\sqrt 2 x_1 -x_3}{2\sqrt 3}\right)^8
\sin \left( \frac {\sqrt 3 x_2 -2\sqrt 2 x_1 +x_3}{2\sqrt 3}\right)^8
\sin \left( \frac {x_2 -\sqrt 3 x_3}{2}\right)^8 \cr
&& \phantom{W(x_{1}, x_{2}, x_{3}):= }\sin \left( \frac {x_2 +\sqrt 3 x_3}{2}\right)^8
\sin \left( \frac {\sqrt 2 x_1 +x_3}{\sqrt 3}\right)^8
\sin(x_2)^8.
\end{eqnarray}
At this point we have to fix the ranges so that $d\mu$ is positive definite. This fixes only the ranges of the parameters which
explicitly appear in the measure. The other ones can be chosen over a period. In this way one covers the whole space an integer
number $N$ of times. This can be checked by comparing the volume obtained by integrating $d\mu$ over the obtained ranges, with the volume
of $E_7/U$ obtained by means of the Macdonald's formula (see \cite{Mac}). Finally one has to identify the finite symmetry group responsible
of the $N$-covering to reduce the ranges over the periods to get the $N=1$ case. Fortunately, in our case most of the work has been
done in \cite{E6} and \cite{F4}. Indeed, the terms $d\mu_{B_{E6}}$, $d\mu_{B_{F_4}}$ $d\mu_{B_{SO(9)}}$ are the same appearing there and
determine the same corresponding ranges of parameters given in that papers. Moreover, $x_1$ does not appears in $d\mu$, so that
its range is the period of $U(1)/\mathbb{Z}_3$ which is $T_g/3=2\sqrt{\frac 23}\pi$. Thus, we are left with the problem to determine the
ranges for the parameters $x_{52}, x_{53}, x_{54}$. To this end, it is convenient to introduce the change of variables
\begin{eqnarray}
&& x= \frac {\sqrt 6 x_{52} -2\sqrt 3 x_{54}}{3}, \qquad y=\frac {\sqrt 3 x_{53} +\sqrt 2 x_{52} +x_{54}}{\sqrt 3},\qquad
z=\frac {\sqrt 3 x_{53} -\sqrt 2 x_{52} -x_{54}}{\sqrt 3},
\end{eqnarray}
with inverse:
\begin{eqnarray}
&& x_{52}= \frac {1}{\sqrt 6} x -\frac{1}{\sqrt 3}z, \qquad x_{53}=\frac 12(y +z),\qquad x_{54}=-\frac 13x +\frac{1}{2\sqrt 3}(y-z).
\end{eqnarray}
Then,
\begin{eqnarray}
W(x_{52}, x_{53}, x_{54})&=&\sin(x)\sin(y)\sin(z)\sin^8(\frac{x-y}{2}) \sin^8(\frac{x+y}{2}) \cr
&&\sin^8(\frac{x-z}{2}) \sin^8(\frac{x+z}{2}) \sin^8(\frac{y-z}{2})
\sin^8(\frac{y+z}{2}),
\end{eqnarray}
which is positive for
\begin{eqnarray}
x\in[0,\pi], \qquad y\in [0,x], \qquad z\in [0,y].
\end{eqnarray}
To check if the range of parameters is the right one, let us compute the integral
$$
{\mathcal I}=\int_R d\mu_{B_{E_7}},
$$
with the determined ranges, to be compared with the volume ${\rm Vol}(E_7/U)$ computed in Appendix \ref{app:macdonald} by means
of the formula of Macdonald. As
$$
dx_{52}dx_{53}dx_{54}= \frac 1{2\sqrt 2} dx dy dz,
$$
we get
\be
{\mathcal I}=\frac {T_g}3 \frac {{\rm Vol}(E_6)}{{\rm Vol}(SO(8))}\frac{1}{2\sqrt{2}} I,
\ee
\begin{align}
I=\!\! \int_0^{\pi}\! dx \int_0^x\!\! dy \int_0^y\!\! dz &\sin(x)\sin(y)\sin(z)\sin^8(\frac{x-y}{2}) \sin^8(\frac{x+y}{2})\cr
&\sin^8(\frac{x-z}{2}) \sin^8(\frac{x+z}{2}) \sin^8(\frac{y-z}{2})
\sin^8(\frac{y+z}{2}).
\end{align}
Using the change of variables
\begin{eqnarray}
\chi= \frac 12 (1-\cos x), \quad\ \xi =\frac 12 (1-\cos y), \quad\ \zeta=\frac 12 (1-\cos z),
\end{eqnarray}
the last integral becomes, see Appendix \ref{app:integral}:
\begin{eqnarray}
I=2^3 \int_0^1 d\chi \int_0^\chi d\xi \int_0^\xi d\zeta (\chi-\xi)^8(\xi-\zeta)^8(\zeta-\chi)^8=
\frac {2}{3^5 \cdot 5 \cdot 11 \cdot 13^2 \cdot 17}. \label{int}
\end{eqnarray}
Since the volumes of $E_6$ and $SO(8)$ can be computed by means of the Macdonald formula giving
\begin{eqnarray}
{\rm Vol}(E_6)=\frac {\sqrt 3\ 2^{17} \pi^{42}}{3^{10} \cdot 5^5 \cdot 7^3 \cdot 11}, \qquad\
{\rm Vol}(SO(8))=\frac {2^{12}\pi^{16}}{3^3\cdot 5},
\end{eqnarray}
we get
\begin{eqnarray}
{\mathcal I}=\frac {2^6 \pi^{27}}{3^{12} \cdot 5^5 \cdot 7^3 \cdot 11^2 \cdot 13^2 \cdot 17},
\end{eqnarray}
which is just twice the volume of $E_7/U$ computed in Appendix \ref{app:macdonald}. However, we can easily see the origin
of this double covering:
consider the element
\begin{eqnarray}
\tau:= -I_{56} \omega^2 =g_1[T_g/2] g_1[2T_g/3]= U(1)[T_g/6], \qquad\ \tau^6=I_{56}.
\end{eqnarray}
One can check that $\tau$ commutes with $E_6$ and with $V$. This means that we can write
\begin{eqnarray}
B e^V U=B\tau^{-1} \tau e^V U =B\tau^{-1} e^V \tau U.
\end{eqnarray}
Since the right $\tau$ can be reabsorbed in the parametrization of $U$, and, on the left, it acts as $\mathbb{Z}_2$ on the factor $U(1)/\mathbb{Z}_3$
in $B$, we see that it identifies the points in the range by means of the relation
$$
x_1 \sim x_1+T_g/6.
$$
Then, to have an injective parametrization, apart from a subset of vanishing measure, we must further restrict the range of
$x_1$ to the interval $x_1\in [0,\sqrt{\frac 23} \pi]$.

\

As a result, we find the following ranges for the $E_7$ Euler parameters:
\begin{align*}
&x_1\in[0,\sqrt{\frac 23}\pi]& & x_2\in[0,2\pi] && x_3\in[0,2\pi] &&
x_4\in[0,2\pi] &&x_5\in[0,\pi] &&
x_6\in[-\frac{\pi}{2},\frac{\pi}{2}]\\
%%%%%%%%%%%%%%%%%%%%%%%%%%%%%%%%%%%%%%%%%%%%%%%%%%%%%
& x_7\in[0,\frac{\pi}{2}]
&&x_8\in[0,\frac{\pi}{2}]&&x_9\in[0,\pi]&&x_{10}\in[0,2\pi]&&x_{11}\in[0,2\pi]
&&x_{12}\in[0,2\pi]\\
%%%%%%%%%%%%%%%%%%%%%%%%%%%%%%%%%%%%%%%%%%%%%%%%%%%%%
&
x_{13}\in[0,\pi]&&x_{14}\in[-\frac{\pi}{2},\frac{\pi}{2}]
&&x_{15}\in[0,\frac{\pi}{2}] &&
x_{16}\in[0,\frac{\pi}{2}]&&x_{17}\in[0,\pi]&&x_{18}\in[0,2\pi]\\
%%%%%%%%%%%%%%%%%%%%%%%%%%%%%%%%%%%%%%%%%%%%%%%%%%%%%
&x_{19}\in[0,2\pi]&&x_{20}\in[0,2\pi]&&x_{21}\in[0,\pi]
&&x_{22}\in[-\frac{\pi}{2},\frac{\pi}{2}]&&x_{23}\in[0,\frac{\pi}{2}]&&x_{24}\in[0,\frac{\pi}{2}]\\
%%%%%%%%%%%%%%%%%%%%%%%%%%%%%%%%%%%%%%%%%%%%%%%%%%%%%
&x_{25}\in[0,\pi]
&&x_{26}\in[0,\pi]&&-\frac{x_{26}}{\sqrt{3}}\leq
x_{27}\leq \frac{x_{26}}{\sqrt{3}}&& x_{28}\in[0,2\pi]
&&x_{29}\in[0,2\pi] &&
x_{30}\in[0,2\pi]\\
%%%%%%%%%%%%%%%%%%%%%%%%%%%%%%%%%%%%%%%%%%%%%%%%%%%%%
&x_{31}\in[0,\pi] &&
x_{32}\in[-\frac{\pi}{2},\frac{\pi}{2}] && x_{33}\in[0,\frac{\pi}{2}]
&&x_{34}\in[0,\frac{\pi}{2}]&&x_{35}\in[0,\pi]
&&x_{36}\in[0,2\pi]\\
%%%%%%%%%%%%%%%%%%%%%%%%%%%%%%%%%%%%%%%%%%%%%%%%%%%%%
&x_{37}\in[0,2\pi]&&
x_{38}\in[0,2\pi]&&
x_{39}\in[0,\pi]&&x_{40}\in[-\frac{\pi}{2},\frac{\pi}{2}]&&
x_{41}\in[0,\frac{\pi}{2}] &&
x_{42}\in[0,\frac{\pi}{2}]\\
%%%%%%%%%%%%%%%%%%%%%%%%%%%%%%%%%%%%%%%%%%%%%%%%%%%%%
&x_{43}\in[0,\pi]&&
 x_{44}\in[0,2\pi] && x_{45}\in[0,2\pi] &&
x_{46}\in[0,2\pi] &&x_{47}\in[0,\pi] &&
x_{48}\in[-\frac{\pi}{2},\frac{\pi}{2}]\\
%%%%%%%%%%%%%%%%%%%%%%%%%%%%%%%%%%%%%%%%%%%%%%%%%%%%%
& x_{49}\in[0,\frac{\pi}{2}]
&&x_{50}\in[0,\frac{\pi}{2}]&& x_{51}\in[0,\pi]
&&x_{55}\in[0,2\sqrt{\frac 23}\pi] &&
 x_{56}\in[0,2\pi] && x_{57}\in[0,2\pi]\\
%%%%%%%%%%%%%%%%%%%%%%%%%%%%%%%%%%%%%%%%%%%%%%%%%%%%%
& x_{58}\in[0,2\pi] &&x_{59}\in[0,\pi]
&& x_{60}\in[-\frac{\pi}{2},\frac{\pi}{2}] && x_{61}\in[0,\frac{\pi}{2}]
&&x_{62}\in[0,\frac{\pi}{2}]&& x_{63}\in[0,\pi]\\
%%%%%%%%%%%%%%%%%%%%%%%%%%%%%%%%%%%%%%%%%%%%%%%%%%%%%
& x_{64}\in[0,2\pi]&& x_{65}\in[0,2\pi]
&&x_{66}\in[0,2\pi]
&& x_{67}\in[0,\pi]&&x_{68}\in[-\frac{\pi}{2},\frac{\pi}{2}]&&
x_{69}\in[0,\frac{\pi}{2}] \\
%%%%%%%%%%%%%%%%%%%%%%%%%%%%%%%%%%%%%%%%%%%%%%%%%%%%%
& x_{70}\in[0,\frac{\pi}{2}]&&x_{71}\in[0,\pi]&& x_{72}\in[0,2\pi]&&
x_{73}\in[0,2\pi]
&&x_{74}\in[0,2\pi]&&x _{75}\in[0,\pi]\\
%%%%%%%%%%%%%%%%%%%%%%%%%%%%%%%%%%%%%%%%%%%%%%%%%%%%%
& x_{76}\in[-\frac{\pi}{2},\frac{\pi}{2}]&&
x_{77}\in[0,\frac{\pi}{2}]&&x_{78}\in[0,\frac{\pi}{2}]&&
x_{79}\in[0,\pi]&& x_{80}\in[0,\pi] && -\frac{x_{80}}{\sqrt{3}}\leq
 x_{81}\leq \frac{x_{80}}{\sqrt{3}}\\
%%%%%%%%%%%%%%%%%%%%%%%%%%%%%%%%%%%%%%%%%%%%%%%%%%%%%
& x_{82}\in[0,2\pi] && x_{83}\in[0,2\pi] &&
x_{84}\in[0,2\pi] &&x_{85}\in[0,\pi] &&
 x_{86}\in[-\frac{\pi}{2},\frac{\pi}{2}] && x_{87}\in[0,\frac{\pi}{2}]
 \\
%%%%%%%%%%%%%%%%%%%%%%%%%%%%%%%%%%%%%%%%%%%%%%%%%%%%%
&x_{88}\in[0,\frac{\pi}{2}]&&x_{89}\in[0,\pi]&&x_{90}\in[0,2\pi]&&x_{91}\in[0,2\pi]
&&x_{92}\in[0,2\pi]&&
x_{93}\in[0,\pi]\\
%%%%%%%%%%%%%%%%%%%%%%%%%%%%%%%%%%%%%%%%%%%%%%%%%%%%%
&x_{94}\in[-\frac{\pi}{2},\frac{\pi}{2}]
&& x_{95}\in[0,\frac{\pi}{2}] &&
 x_{96}\in[0,\frac{\pi}{2}]&& x_{97}\in[0,\pi] &&
x_{98}\in[0,2\pi]& & x_{99}\in[0,2\pi] \\
%%%%%%%%%%%%%%%%%%%%%%%%%%%%%%%%%%%%%%%%%%%%%%%%%%%%%
& x_{100}\in[0,2\pi] &&x_{101}\in[0,\pi]
&& x_{102}\in[-\frac{\pi}{2},\frac{\pi}{2}] && x_{103}\in[0,\frac{\pi}{2}]
&&x_{104}\in[0,\frac{\pi}{2}]&&x_{105}\in[0,\pi]\\
%%%%%%%%%%%%%%%%%%%%%%%%%%%%%%%%%%%%%%%%%%%%%%%%%%%%%
 & x_{106}\in[0,2\pi] && x_{107}\in[0,2\pi] &&
x_{108}\in[0,2\pi]
&&x_{109}\in[0,\pi] &&
 x_{110}\in[-\frac{\pi}{2},\frac{\pi}{2}] &&
 x_{111}\in[0,\frac{\pi}{2}] \\
%%%%%%%%%%%%%%%%%%%%%%%%%%%%%%%%%%%%%%%%%%%%%%%%%%%%%
&x_{112}\in[0,\frac{\pi}{2}]&&
x_{113}\in[0,2\pi]&& x_{114}\in[0,\pi]&& x_{115}\in[-\frac{\pi}{2},\frac{\pi}{2}]& & x_{116}\in[-\frac{\pi}{2},\frac{\pi}{2}]&&
x_{117}\in[-\frac{\pi}{2},\frac{\pi}{2}] \\
%%%%%%%%%%%%%%%%%%%%%%%%%%%%%%%%%%%%%%%%%%%%%%%%%%%%%
& x_{118}\in[0,\pi] &&
 x_{119}\in[0,2\pi]&& x_{120}\in[0,\pi]&& x_{121}\in[-\frac{\pi}{2},\frac{\pi}{2}]
&&x_{122}\in[-\frac{\pi}{2},\frac{\pi}{2}]
&&x_{123}\in[0,\pi] \\
%%%%%%%%%%%%%%%%%%%%%%%%%%%%%%%%%%%%%%%%%%%%%%%%%%%%%
& x_{124}\in[0,2\pi]&& x_{125}\in[0,\pi]
&&x_{126}\in[-\frac{\pi}{2},\frac{\pi}{2}] &&x_{127}\in[0,\pi] &&
x_{128}\in[0,2\pi]&& x_{129}\in[0,\pi]\\
%%%%%%%%%%%%%%%%%%%%%%%%%%%%%%%%%%%%%%%%%%%%%%%%%%%%%
&x_{130}\in[0,\pi] &&
 x_{131 }\in[0,2\pi]&& x_{132}\in[0,\pi]
&&x_{133}\in[0,4\pi]
\end{align*}
and
\begin{align*}
& x_{52}= \frac {1}{\sqrt 6} x -\frac{1}{\sqrt 3}z, && x_{53}=\frac
12(y +z), && x_{54}=-\frac 13x +\frac{1}{2\sqrt 3}(y-z),
\end{align*}
with $x\in[0,\pi]$, $y\in[0,x]$ and $z\in[0,y]$.

%%%%%%%%%%%%%%%%%%%%%%%%%%%%%%%%%%%%%%%%%%%%%%%%%%%%%%%%%%%%%%%%%%%%%%%%%%%%%%%%%%%%%%%%%%%%%%%%%%%%%%%%%%%%%%%%%%%%%%%
\section{Deduction of the measure}\label{app:measure}
The invariant measure on $E_7/U$ is
\begin{equation}
d\mu_{B_{E_7}}=|\det J_\p^\perp|,
\end{equation}
where $J_\p^\perp$ is the projection of
$$
J_\p=e^{-V} B^{-1} d(B e^V)=dV +e^{-V} B^{-1} dB e^V
$$
on the real linear subspace of Lie($E_7$) spanned by $Y_2,Y_3,Y_{82},\ldots,Y_{133}$.
More specifically we have
\begin{align*}
J_\p=&dV+ e^{-V}U_1^{-1}dU_1e^V+e^{-V}B_{SO(9)}^{-1}d B_{SO(9)} e^V +e^{-V}B_{SO(9)}^{-1}B_{F_4}^{-1}d B_{F_4} B_{SO(9)} e^V \cr
&+e^{-V}B_{SO(9)}^{-1}B_{F_4}^{-1}B_{E_6}^{-1}d B_{E_6} B_{F_4} B_{SO(9)} e^V,
\end{align*}
where we used the fact that the $U(1)$ factor commutes with $E_6$. To compute $\det J_\p^\perp$ we can then proceed by analyzing
the summands term by term (here we will set $H:={\rm Lie}(U)$):
\begin{itemize}
\item
$dV=Y_2dx_2+Y_{82}dx_{82}+Y_{99}dx_{99}$, which has obviously non zero
projection just on the matrices $Y_2$, $Y_{82}$ and $Y_{99}$. Thus, the remaining terms have to be projected on
$Y_3, Y_{83},\ldots, Y_{98}, Y_{100},\ldots, Y_{133}$.
\item
The $SO(9)$ current $J_{B_{SO(9)}}:=B_{SO(9)}^{-1}d B_{SO(9)}$ can be
split in two orthogonal pieces $J_{B_{SO(9)}}=J_{SO(8)}\oplus J_{SO(9)\backslash SO(8)}$, where $J_{SO(8)}$ is the projection
of $J_{B_{SO(9)}}$ into Lie($SO(8)$). The space $V$ commute with this $SO(8)$, thus
$e^{-V}J_{SO(8)}e^{V}=J_{SO(8)}$ and as $SO(8)\subset E_6$ it has vanishing projection over $H^\perp$. Then $e^{-V}B_{SO(9)}^{-1}d
B_{SO(9)} e^V|_{H^\perp}=e^{-V}J_{SO(9)\backslash SO(8)}e^V|_{H^\perp}$. Thus, this term will contribute to the determinant with the
term
$$
A=\det e^{-V}J_{SO(9)\backslash SO(8)}e^V|_{H^\perp}.
$$
To compute it we first note that $J_{SO(9)\backslash SO(8)}=\sum_{a=1}^8 J_{B_{SO(9)}}^a Y_{a+47}$
and that $Ad_{e^{-V}} Y_{a+47}|_{H^\perp}$, $a=1,\ldots,8$ have non vanishing components only on the subspace generated by
$Y_{a+99}$, $a=1,\ldots, 8$. This means that if we define the eight dimensional spaces
\begin{eqnarray}
W_{47}:=\sum_{a=1}^8 \mathbb{R} Y_{a+47}, \qquad\ W_{99}:=\sum_{a=1}^8 \mathbb{R} Y_{a+99},
\end{eqnarray}
define the linear map
\begin{eqnarray}
\rho: W_{47} \longrightarrow W_{99},\ Y\longmapsto e^{-V} Y e^{V}|_{H^\perp}
\end{eqnarray}
and define the matrix $R$ associated to $\rho$ by the bases $\{ Y_{a+n} \}_{a=1}^8$, $n=47,99$ of $W_n$ respectively, then we can
write $A=\det R \det J_{SO(9)\backslash SO(8)}$.
Now, $\det R$ can be easily computed by means of a computer, whereas $\det J_{SO(9)\backslash SO(8)}=d\mu_{B_{SO(9)}}$ by construction
and has yet been computed in \cite{F4}. The result is then
\begin{eqnarray}
A=\sin^8\left(\frac {x_2-\sqrt 3 x_3}2\right) d\mu_{B_{SO(9)}}.
\end{eqnarray}
Moreover, the remaining terms have to be projected on
$H_1^\perp :={\rm Span}\{Y_3, Y_{83},\ldots, Y_{98}, Y_{108},\ldots, Y_{133}\}_{\mathbb{R}}$.
\item
Let us now consider the term $e^{-V}B_{SO(9)}^{-1}B_{F_4}^{-1}d B_{F_4} B_{SO(9)} e^V$. The $F_4$ current $J_{B_{F_4}}:=B_{F_4}^{-1}d B_{F_4}$
can be split in two orthogonal parts $J_{B_{F_4}}=J_{SO(9)}\oplus J_{F_4\backslash SO(9)}$, where $J_{SO(9)}$ is the projection
of $J_{B_{F_4}}$ over Lie($SO(9)$). The adjoint action of $SO(9)$ on $J_{SO(9)}$ has value in Lie($SO(9)$) which under $Ad_{e^{-V}}$
has vanishing projection over $H_1^\perp$, as seen before. On the other hand,
$$
\Phi:=Ad_{SO(9)} : {\rm Lie}(F_4/SO(9))\longrightarrow {\rm Lie}(F_4/SO(9))
$$
acts as an orthogonal map so that we have
\begin{align*}
B:=&\det (e^{-V}B_{SO(9)}^{-1}B_{F_4}^{-1}d B_{F_4} B_{SO(9)} e^V)|_{H_1^\perp}=\det(e^{-V}\Phi(J_{F_4\backslash SO(9)}) e^V)|_{H_1^\perp}\\
=&\det \Phi \ \det(e^{-V}J_{F_4\backslash SO(9)} e^V)|_{H_1^\perp}=\det(e^{-V}J_{F_4\backslash SO(9)} e^V)|_{H_1^\perp}.
\end{align*}
Now we can proceed as in the previous case: first, one can check that
\eqn
J_{F_4\backslash SO(9)}\in W_{24,39}= {\rm Span}\{Y_{24+a}, Y_{39+a}, a=1,\ldots,8 \}_{\mathbb{R}}.
\feqn
Then, one sees that
$$
Ad_{e^{-V}} (W_{24,39})|_{H_1^\perp}\subseteq W_{82}:={\rm Span}\{Y_{82+a}, a=1,\ldots,16 \}_{\mathbb{R}},
$$
so that if we define the linear map
\begin{eqnarray}
\rho' W_{24,39} \longrightarrow W_{82},\ Y\longmapsto e^{-V} Y e^{V}|_{H_1^\perp},
\end{eqnarray}
and $M'$ is the associated matrix, then we get
\begin{eqnarray}
B=\det M' \det J_{F_4\backslash SO(9)}=\sin^8 x_2 \sin^8 \left(\frac {x_2+\sqrt 3 x_3}2\right) d\mu_{B_{F_4}},
\end{eqnarray}
where $d\mu_{B_{F_4}}$ has been computed in \cite{F4}.
The remaining terms in $J_\p$ must be projected on $H_2^\perp :={\rm Span}\{Y_3, Y_{108},\ldots, Y_{133}\}_{\mathbb{R}}$.
\item The computation we have to consider is the contribution of the terms
\begin{eqnarray*}
e^{-V}U_1^{-1}dU_1e^V+e^{-V}B_{SO(9)}^{-1}B_{F_4}^{-1}B_{E_6}^{-1}d B_{E_6} B_{F_4} B_{SO(9)} e^V,
\end{eqnarray*}
which we rewrite conveniently in the form
\begin{eqnarray*}
e^{-V}B_{SO(9)}^{-1}B_{F_4}^{-1}(J_{U(1)}\oplus J_{B_{E_6}}) B_{F_4} B_{SO(9)} e^V:=
e^{-V}B_{SO(9)}^{-1}B_{F_4}^{-1}(U_1^{-1}dU_1+ B_{E_6}^{-1}d B_{E_6}) B_{F_4} B_{SO(9)} e^V.
\end{eqnarray*}
As before, $J_{B_{E_6}}$ can be split in two orthogonal parts as $J_{B_{E_6}}=J_{F_4}\oplus J_{E_6\backslash F_4}$, where $J_{F_4}$ is the projection
of $J_{B_{E_6}}$ over Lie($F_4$), and by construction
$$
e^{-V}B_{SO(9)}^{-1}B_{F_4}^{-1} J_{F_4} B_{F_4} B_{SO(9)} e^V|_{H_2^\perp}=0.
$$
Moreover, $Ad_{F_4}$ (and then $Ad_{SO(9)}$) acts as an orthogonal map on the space
$$
W_{1,55}:={\rm Span}\{Y_{1}, Y_{55+a}, a=1,\ldots, 26 \}_{\mathbb{R}}
$$
so that, with the same arguments as before, we can restrict to compute the term
\begin{eqnarray}
C:=\det(e^{-V}(U_1^{-1}dU_1\oplus J_{E_6\backslash F_4}) e^V|_{H_2^\perp}).
\end{eqnarray}
If we construct the map
\begin{eqnarray}
\rho'': W_{1,55}\longrightarrow H_2^\perp, \ Y\longmapsto e^{-V} Y e^{V}|_{H_2^\perp}
\end{eqnarray}
and the natural associated matrix $M''$, then we get
\begin{eqnarray}
&& C=\det M'' dx_1 \det J_{E_6\backslash F_4}=
\frac 1{2^9} \left( \cos (2{x_2})-\cos \frac {2\sqrt 2 x_1 +2x_3}{\sqrt 3} \right)
\sin \frac {\sqrt 6 x_1-2\sqrt 3 x_3}3 \cdot  \cr
&& \phantom{C=\det M" dx_1 \det J_{E_6\backslash F_4}=} \cdot \left( \cos x_2 -\cos \frac {2\sqrt 6 x_1-\sqrt 3 x_3}3  \right) dx_1 d\mu_{B_{E_6}},
\end{eqnarray}
where $d\mu_{B_{E_6}}$ is the measure computed in \cite{E6}.
\end{itemize}
Collecting all the terms, we finally get $d\mu_{B_{E_7}}=ABC$, which after application of the formulas of prosthapheresis give
the result (\ref{measure}).

%%%%%%%%%%%%%%%%%%%%%%%%%%%%%%%%%%%%%%%%%%%%%%%%%%%%%%%%%%%%%%%%%%%%%%%%%%%%%%%%%%%%%%%%
\section{Macdonald formulas} \label{app:macdonald}
We can compute the volume of the group by means of the Macdonald formula, see \cite{Mac}. We are using the invariant measure induced
by an invariant scalar product on the algebra. The Cartan subalgebra $\mathcal{C}$ is generated by the matrices
$Y_1, Y_4, Y_9, Y_{18}, Y_{39}, Y_{56}, Y_{73}$. These matrices are orthonormal, so that the same holds true for the dual basis.
In this normalization the roots have length $\sqrt 2$. In particular, by diagonalizing $ad_{Y_i}$, $i=1,4,9,18,39,56,73$ simultaneously,
one sees that any choice of simple roots can be written in the form
\begin{eqnarray*}
&& \alpha_1= \frac 12 (\lambda_1-\lambda_2-\lambda_3-\lambda_4-\lambda_5-\lambda_6+\sqrt 2 \lambda_7),\\
&& \alpha_2=\lambda_1+\lambda_2,\\
&& \alpha_3=\lambda_2-\lambda_1,\\
&& \alpha_4=\lambda_3-\lambda_2,\\
&& \alpha_5=\lambda_4-\lambda_3,\\
&& \alpha_6=\lambda_5-\lambda_4,\\
&& \alpha_7=\lambda_6-\lambda_5,
\end{eqnarray*}
where $\lambda_j$ is an orthonormal basis (w.r.t. the product induced by the one on $\mathcal{C}$).
The coroots $\alpha_i^{\vee}$ then coincide with the roots, and the fundamental region, that is the fundamental torus generated on
$\mathcal{C}$ by the root lattice, has then volume
$$
V_T=|\alpha_1\wedge \ldots \wedge \alpha_7|=\sqrt 2.
$$
Moreover, the rational cohomology of $E_7$ is same of the product of seven spheres \cite{chevalley}, \cite{Borel-ch}:
$$
H(E_7;\mathbb{Q})=H(S^3\times S^{11}\times S^{15}\times S^{19}\times S^{23}\times S^{27}\times S^{35};\mathbb{Q}).
$$
Applying the formula of Macdonald, we get
\begin{eqnarray}
{\rm Vol}(E_7)=V_T \prod_{i=1}^7 {\rm Vol}(S^{d_i}) \prod_{\alpha\neq 0} |\alpha^\vee|=
\frac {\sqrt 2 \cdot 2^{23} \pi^{70}}{3^{22}\cdot 5^{10}\cdot 7^6 \cdot 11^3\cdot 13^2 \cdot 17},
\end{eqnarray}
where $S^{d_i}$ are the spheres appearing in the cohomology, and the second product is over all non vanishing roots.\\
In a similar way, we can compute the volume of the subgroups needed in the computations. In particular:
\begin{eqnarray}
&& {\rm Vol}(E_6)=\frac {\sqrt 3\ 2^{17} \pi^{42}}{3^{10} \cdot 5^5 \cdot 7^3 \cdot 11},\\
&& {\rm Vol}(U)={\rm Vol}(E_6) {\rm Vol}(U(1))/3=\frac {\sqrt 2\ 2^{18} \pi^{43}}{3^{10} \cdot 5^5 \cdot 7^3 \cdot 11},
\end{eqnarray}
and then
\begin{eqnarray}
{\rm Vol}(E_7/U)=\frac {2^5 \pi^{27}}{3^{12} \cdot 5^5 \cdot 7^3 \cdot 11^2 \cdot 13^2 \cdot 17}.
\end{eqnarray}

%%%%%%%%%%%%%%%%%%%%%%%%%%%%%%%%%%%%%%%%%%%%%%%%%%%%%%%%%%%%%%%%%%%%%%%%%%%%%%%%%%%%%%%%%%%%%%%%%%%%%%%%%%%%%%%

\section{Computation of the integral} \label{app:integral}
To compute the integral $I$ in (\ref{int}), we here compute the more general integral
\begin{eqnarray}
I(a,b,c)=\int_0^1 dx \int_0^x dy \int_0^y dz (x-y)^{a-1}(y-z)^{b-1}(x-z)^{c-1},
\end{eqnarray}
and we will get $I$ from the the relation $I=2^3 I(9,9,9)$. To this hand we will use the following useful representation
of the hypergeometric function:
\begin{eqnarray}
_2F_1(\alpha,\beta;\gamma;z)\equiv F(\alpha,\beta;\gamma;z) =\frac {\Gamma (\gamma)}{\Gamma (d) \Gamma (\gamma-d)}
\int_0^1 dt\ t^{d-1} (1-t)^{\gamma-d-1} F(\alpha,\beta; d; zt), \label{hyp1}
\end{eqnarray}
which for $d=\alpha$ takes the form
\begin{eqnarray}
F(\alpha,\beta;\gamma;z)= \frac {\Gamma (\gamma)}{\Gamma (\alpha) \Gamma (\gamma-\alpha)} \int_0^1 dt\ t^{\alpha-1} (1-t)^{\gamma-\alpha-1}
(1-zt)^{-\beta}. \label{hyp2}
\end{eqnarray}
First, let us change the variables by introducing the new coordinates $(x,s,t)$ such that $(x,y,z)=(x,xs,xst)$. This gives
\begin{eqnarray}
I(a,b,c)=\frac 1{a+b+c} \int_0^1 ds \int_0^1 dt s^b(1-s)^{a-1}(1-t)^{b-1}(1-st)^{c-1}.
\end{eqnarray}
Next we can first integrate over $t$. By using (\ref{hyp2}), we get
\begin{eqnarray}
&&I(a,b,c)=\frac 1b \frac 1{a+b+c} \int_0^1 ds\ s^b(1-s)^{a-1} F(1,1-c;b+1;s)\cr
&&\phantom{I(a,b,c)}=\frac 1b \frac 1{a+b+c} \frac {\Gamma(a)\Gamma(b+1)}{\Gamma(a+b+1)}
\left[\frac {\Gamma(a+b+1)}{\Gamma(a)\Gamma(b+1)} \int_0^1 ds\ s^b(1-s)^{a-1} F(1,1-c;b+1;s)\right].
\end{eqnarray}
The therm in the square brackets has exactly the form (\ref{hyp1}), with $\alpha=1,\ \beta=1-c,\ d=b+1,\ \gamma=a+b+1,\ z=1$.
Then we can write
\begin{eqnarray}
I(a,b,c)=\frac 1b \frac 1{a+b+c} \frac {\Gamma(a)\Gamma(b+1)}{\Gamma(a+b+1)} F(1, 1-c; a+b+1; 1).
\end{eqnarray}
Finally, by using
\begin{eqnarray}
F(\alpha,\beta;\gamma;1)= \frac {\Gamma (\gamma) \Gamma(\gamma-\alpha-\beta)}{\Gamma (\gamma-\alpha) \Gamma (\gamma-\beta)},
\end{eqnarray}
and the property $z\Gamma(z)=\Gamma(z+1)$, we find
\begin{eqnarray}
I(a,b,c)=\frac 1{a+b+c}\ \frac 1{a+b+c-1}\ \frac {\Gamma(a)\Gamma(b)}{\Gamma(a+b)}.
\end{eqnarray}
In particular:
\begin{eqnarray}
I=2^3 I(9,9,9)=\frac {2^2}{3^3 \cdot 13} \frac {(8!)^2}{17!}=\frac {2}{3^5 \cdot 5 \cdot 11 \cdot 13^2 \cdot 17}.
\end{eqnarray}
%%%%%%%%%%%%%%%%%%%%%%%%%%%%%%%%%%%%%%%%%%%%%%%%%%%%%%%%%%%%%%%%%%%%%%%%%%%%%%%%%%%%%%%%%%%%%%%%%%%%%%%%%%

\section{Roots and range of parameters for the $E_{7(7)}$ construction} \label{app:range77}
To our aim we need to perform a choice of positive roots w.r.t. the Cartan subalgebra $H=<D_\alpha>_{\mathbb{R}}$.
First, we can write
\begin{eqnarray*}
&& L=H\oplus <J^+>_{\mathbb R} \oplus <J^->_{\mathbb R}:= H\oplus <\{ J^+_{kl}=\frac 1{\sqrt 2} (-iS_{kl}+A_{kl})\, |\, k<l  \}>_{\mathbb R}
\oplus <\{J^-_{kl}=\frac 1{\sqrt 2} (-iS_{kl}-A_{kl})\, |\, k<l\}>_{\mathbb R},\\
&& \lambda^4:=<{\mathcal J}>_{\mathbb R}:= <\{ {\mathcal J}_I=\lambda_I| I\in {\mathcal I} \}>_{\mathbb R}.
\end{eqnarray*}
where ${\mathcal I}$ is the set of 4/indices.
\begin{prop}
The set $J^+ \cup J^- \cup {\mathcal J}$ diagonalizes simultaneously the adjoint action of $H$.
\end{prop}
\begin{proof}
By direct computation of the action of $[D_\alpha, J^\pm_{kl}]$ on $e_{ij}$ and on $\varepsilon^{ij}$ we get:
\begin{eqnarray}
[D_\alpha, J^\pm_{kl}]=\pm i(D_\alpha^k-D_\alpha^l) J^\pm_{kl},
\end{eqnarray}
where for ${\mathcal J}_{I}=\lambda_{i_1i_2i_3i_4}$, using $\tr D_\alpha=0$ we get
\begin{eqnarray}
[D_\alpha, {\mathcal J}_I]=i(D_\alpha^{i_1}+D_\alpha^{i_2}+D_\alpha^{i_3}+D_\alpha^{i_4}) {\mathcal J}_I.
\end{eqnarray}
\end{proof}
We now fix an explicit choice for a basis of $H$, suitable for our purposes:
\begin{eqnarray*}
&& D_1=\frac 1{\sqrt 2}\diag\{1,-1,-1,1,0,0,0,0\}, \qquad D_2=\frac 1{\sqrt 2}\diag\{1,-1,1,-1,0,0,0,0\},
\qquad D_3=\frac 1{\sqrt 2}\diag\{1,1,-1,-1,0,0,0,0\}, \\
&& D_4=\frac 1{\sqrt 2}\diag\{0,0,0,0,1,-1,-1,1\}, \qquad D_5=\frac 1{\sqrt 2}\diag\{0,0,0,0,1,-1,1,-1\},
\qquad D_6=\frac 1{\sqrt 2}\diag\{0,0,0,0,1,1,-1,-1\}, \\
&& D_7=\frac 12\diag\{1,1,1,1,-1,-1,-1,-1\}.
\end{eqnarray*}
Before continuing, let us recall the structure of the roots for the $E_7$ type algebras. Let $R^+$ the set of positive roots
and $L_i$ an orthonormal basis for $H^*_{\mathbb{R}}$ the real space spanned by all roots. Then, the positive roots system for an
$E_7$ Lie algebra is (see \cite{fulton-harris}, chapter 21, p.333)
\begin{eqnarray}
R^+=\{L_j+L_i\}_{i<j\leq6}\cup \{L_j+L_i\}_{i<j\leq6} \cup \{\sqrt 2 L_7 \}\cup\{\frac {\pm L_1 \pm \ldots \pm L_6
+\sqrt 2 L_7}2\}_{\mbox{odd number of sign -}} \label{roots}
\end{eqnarray}
and in particular a choice of simple roots is
\begin{eqnarray}
&& \alpha_1=\frac {L_1-L_2-L_3-L_4-L_5-L_6+\sqrt 2 L_7}2, \qquad \alpha_2=L_1+L_2, \qquad \alpha_3 =L_2-L_1, \cr
&& \alpha_4 =L_3-L_2, \qquad  \alpha_5 =L_4-L_3, \qquad \alpha_6 =L_5-L_4, \qquad \alpha_7 =L_6-L_5.
\end{eqnarray}
We normalized the basis for $H$ so that $D_\alpha \cdot D_\beta =2\delta_{\alpha\beta}$. An orthonormal basis for $H^*_{\mathbb{R}}$
is thus provided by $L_\alpha (H_\beta)=\sqrt{2} \delta_{\alpha\beta}$. Let us then introduce the subset ${\mathcal I}_0 \subset {\mathcal I}$
such that ${\mathcal I}={\mathcal I}_0 \cup \tilde{\mathcal I}_0$ defined as follows:
\begin{eqnarray}
{\mathcal I}_0=\{I\in {\mathcal I}: i_1,i_2,i_3 \in \{1,2,3,4,5\}\}.
\end{eqnarray}
In other words it is the set of ordered tetra-indices such that or all indices run from 1 to 5, or $i_7\in {6,7,8}$. Its cardinality is then
$\binom54+3\binom53=35$. Set
\begin{equation}
{\mathcal J}={\mathcal J}^+ \cup {\mathcal J}^-:=\{\lambda_I \in {\mathcal J}: I\in {\mathcal I}_0\}\cup
\{\lambda_I \in {\mathcal J}: \tilde I\in {\mathcal I}_0, \epsilon_{I,\tilde I}=1 \}.
\end{equation}
\begin{prop}
The set $J^+\cup {\mathcal J}^+$ consists of all eigenvectors associated to all positive roots of $\lie E_{7}$. The corresponding
roots are
\begin{eqnarray}
\{ \beta_{kl}:= \sum_{\alpha=1}^7 \frac {D_\alpha^k -D_\alpha^l}{\sqrt 2} L_\alpha\}_{k<l} \cup
\{ \beta_{i_1i_2i_3i_4}:=\sum_{\alpha=1}^7 \frac {D_\alpha^{i_1} +D_\alpha^{i_2}+D_\alpha^{i_3} +D_\alpha^{i_4}}{\sqrt 2} L_\alpha  \}_{i_1i_2i_3i_4\in
{\mathcal I}_0}.
\end{eqnarray}
In particular the simple roots are
\begin{eqnarray}
\alpha_1=\beta_{45}, \qquad \alpha_2= \beta_{12}, \qquad \alpha_3=\beta_{34}, \qquad \alpha_4=\beta_{23}, \qquad \alpha_5=\beta_{3458},
\qquad \alpha_6=\beta_{78}, \qquad \alpha_7=\beta_{67}.
\end{eqnarray}
\end{prop}
The proof is by direct inspection. Note that the order for the simple roots in the theorem is the same as in Fig. \ref{fig}.
This result is very helpful for computing the function $f={\rm det} [\Pi \circ {\rm Ad}_{e^{-V}} : \mathfrak{u} \rightarrow \mathfrak {t} ]$.
Indeed, from both Propositions we see that a basis for $\mathfrak u$ is given by the matrices of the form
$$S=i\frac {J^++J^-}{\sqrt 2}, $$ whereas a basis for $\mathfrak t$ is given by the elements of the form
$$S=\frac {J^+-J^-}{\sqrt 2}.$$ Using this and
$$
{\rm Ad}_{e^{\sum_a y^a D_a}} J_\beta^\pm ={e^{\pm\sum_a y^a \beta(D_a)}} J_\beta^\pm
$$
for a given root $\beta$, one finally obtains
\begin{eqnarray}
  |f(\vec y)|=\prod_{\beta \in {\rm Rad}^+} \sin (|\sum_{a=1}^7 y^a \beta(D_a)| ),
\end{eqnarray}
where ${\rm Rad}^+$ is the set of positive roots w.r.t. $V$.

%%%%%%%%%%%%%%%%%%%%%%%%%%%%%%%%%%%%%%%%%%%%%%%%%%%%%%%%%%%%%%%%%%%%%%%%%%%%%%%%%%%%%%%%%%%%%%%%%%%
\section{Some details for the $E_{7(-5)}$ construction}\label{app:ultimissima/e/basta}
Here we will specify the subgroup $K$ of $U_5$ commuting with the torus $e^{H_4}$. To this end start we start by looking for the
subalgebra $\mathfrak k$ of $\mathfrak{u}_5$ commuting with $H_4$. This can be done by means of Mathematica and gives
\begin{eqnarray}
\mathfrak{k}=\langle M_1,\ldots, M_9 \rangle_{\mathbb R},
\end{eqnarray}
with
\begin{eqnarray}
&& M_1= \frac 12(L_1-L_9+L_{21}-L_{25}), \\
&& M_2= \frac 12(L_2+L_8+L_{14}-L_{28}), \\
&& M_3= \frac 12(L_3-L_7+L_{19}+L_{27}), \\
&& M_4= \frac 12(L_4-L_{13}+L_{20}+L_{23}), \\
&& M_5= \frac 12(L_5-L_{15}+L_{18}-L_{22}), \\
&& M_6= \frac 12(L_{10}-L_{12}+L_{16}+L_{24}), \\
&& M_7= \frac 1{\sqrt 2}(L_{45}+L_{46}), \\
&& M_8= L_{49}, \\
&& M_9= L_{50}.
\end{eqnarray}
These generate an algebra ${\rm so}(4) \oplus {\rm su}(2)$ whose exponentiation gives the group $K^0 ={\rm SO}(4) \times {\rm SU}(2)$.
A general analysis, which will be presented in \cite{to-appear}, shows that $K$ contains an extra $\mathbb Z_2^4$ factor so that
\begin{eqnarray}
K={\rm SO}(4) \times {\rm SU}(2)\times \mathbb Z_2^4.
\end{eqnarray}

%%%%%%%%%%%%%%%%%%%%%%%%%%%%%%%%%%%%%%%%%%%%%%%%%%%%%%%%%%%%%%%%%%%%%%%%%%%%%%%%%%%%%%%%%%%%%%%%%%%

\end{appendix}

\end{document}